\documentclass[a4paper, oneside]{article}
\usepackage[usenames, dvipsnames]{color}
\usepackage{booktabs}
\usepackage[multiple,symbol]{footmisc}
\definecolor{light-gray}{gray}{0.85}
\usepackage{graphicx}
\usepackage{tikz}
\usetikzlibrary{arrows.meta}
\usetikzlibrary{calc}
\usetikzlibrary{shapes.misc}
\usepackage{enumitem}
\usepackage{caption}
\usepackage{subcaption}
\usepackage{changepage}
\usepackage[draft]{hyperref}
\newcommand*{\Cdot}{\raisebox{-0.25ex}{\scalebox{1.25}{$\cdot$}}}
\usepackage[english]{babel}
\usepackage[utf8]{inputenc}
\usepackage{amssymb,amsmath, amsthm, mathtools}
\theoremstyle{Theorem}
\newtheorem{lemma}{Lemma}[section]
\newtheorem{prop}{Proposition}[section]
\newtheorem{theorem}{Theorem}[section]
\newtheorem{cor}{Corollary}[section]

\DeclareMathOperator*{\argmax}{arg\,max}
\DeclareMathOperator*{\argmin}{arg\,min}
\newtheorem{definition}{Definition}[section]
\newtheorem{example}{Example}
\usepackage{thmtools}

\theoremstyle{definition}

\definecolor{hall}{gray}{0.6}
\usepackage{fullpage}
\DeclarePairedDelimiter{\ceil}{\lceil}{\rceil}

\usepackage[square,numbers]{natbib}
\setlist[description]{font=\normalfont}
\usetikzlibrary{external}
\usetikzlibrary{arrows,automata}
\usepackage{kbordermatrix}
\usepackage{soul}
\def\X{{\cal X}}
\def\Y{{\cal Y}}
\def\Z{{\cal Z}}
\def\B{{\cal B}}
\DeclarePairedDelimiter\floor{\lfloor}{\rfloor}
\def\E{{\mathbb{E}}}
\def\DEF{\stackrel{\mbox{\scriptsize{\normalfont{def}}}}{=}}
\makeatletter
\newsavebox\myboxA
\newsavebox\myboxB
\newlength\mylenA
\usepackage{titlesec}
\usepackage[titletoc,toc,title]{appendix}
\usepackage{colonequals}
\makeatletter
\def\namedlabel#1#2{\begingroup
   \def\@currentlabel{#2}%
   \label{#1}\endgroup
}
\makeatother
\mathtoolsset{centercolon}
\usepackage{accents}

\usepackage{authblk}
\begin{document}
\title{Regenerativity of Viterbi process for pairwise Markov models\footnote{This is a post-peer-review, pre-copyedit version of an article published in Journal of Theoretical Probability. The final authenticated version is available online at: \url{http://dx.doi.org/10.1007/s10959-020-01022-z}} \footnote{The research is supported by Estonian  institutional research funding
IUT34-5 and PRG865.}}
\date{}
\author{Jüri Lember}
\author{Joonas Sova}
\affil{\small University of Tartu, Narva  mnt 18, Tartu, Estonia.\\
Email: \textit{jyril@ut.ee}; \textit{joonas.sova@ut.ee}}
\maketitle

\begin{abstract}\noindent For hidden Markov models one of the most popular estimates of the hidden chain is the Viterbi path -- the path maximising the posterior probability. We consider a more general setting, called the pairwise Markov model (PMM), where the joint process consisting of finite-state hidden process and observation process is assumed to be a Markov chain. It has been recently proven that under some conditions the Viterbi path of the PMM can almost surely be extended to infinity, thereby defining the infinite Viterbi decoding of the observation sequence, called the Viterbi process. This was done by constructing a block of observations, called a barrier, which ensures that the Viterbi path goes trough a given state whenever this block occurs in the observation sequence. In this paper we prove that the joint process consisting of Viterbi process and PMM is regenerative. The proof involves a delicate construction of regeneration times which coincide with the occurrences of barriers. As one possible application of our theory, some results on the asymptotics of the Viterbi training algorithm are derived.\end{abstract}

\section{Introduction and preliminaries}
\subsection{Introduction}
We consider a Markov chain $Z =\{Z_k \}_{k\geq 1}$ with product
state space $\mathcal{X}\times \mathcal{Y}$, where $\mathcal{Y}$ is
a finite set (state space) and $\mathcal{X}$ is an arbitrary
separable metric space (observation space). Thus, the process $Z$
decomposes as $Z=(X,Y)$, where $X=\{X_k \}_{k\geq 1}$ and $Y=\{Y_k
\}_{k\geq 1}$ are random processes taking values in  $\mathcal{X}$
and $\mathcal{Y}$, respectively. The process $X$ is identified as an
observation process and the process $Y$, sometimes called the {\it
regime}, models  the observations-driving hidden state sequence.
Therefore our general model contains many well-known stochastic
models as a special case: hidden Markov models (HMM), Markov
switching models, hidden Markov models with dependent noise and many
more.  The {\it segmentation} or {\it path estimation} problem
consists of estimating the realization of $(Y_1,\ldots,Y_n)$ given a
realization $x_{1:n}$ of $(X_1,\ldots,X_n)$. A standard estimate is
any path $v_{1:n}\in \mathcal{Y}^n$ having maximum posterior
probability:
$$v_{1:n}=\argmax_{y_{1:n}}P(Y_{1:n}=y_{1:n}|X_{1:n}=x_{1:n}).$$
Any such  path is called {\it Viterbi path} and we are interested in
the behaviour of $v_{1:n}$ as $n$ grows. The study of asymptotics of
Viterbi path is complicated by the fact that adding one more
observation, $x_{n+1}$ can change the whole path, and so it is not
clear, whether there exists a limiting infinite Viterbi path.

It has been recently proven \cite{PMMinf} that under some conditions the infinite Viterbi path exists
for almost every realization $x_{1:\infty}$ of $X$, allowing to define an infinite Viterbi decoding of $X$, called the \textit{Viterbi process.} This was done trough construction of \textit{barriers}. A barrier is a fixed-sized block in the observations $x_{1:n}$ that  fixes the Viterbi path up to
 itself: for every continuation of $x_{1:n}$, the Viterbi path up to
 the barrier remains unchanged. Therefore, if
almost every realization of $X$-process contains
infinitely many barriers, then the infinite Viterbi path exists
a.s.

Having infinitely many barriers is not necessary for
existence of infinite Viterbi path (see \cite[Example 1.2]{PMMinf}), but the
barrier-construction has several advantages. One of them is that it
allows to construct the infinite path {\it piecewise}, meaning that
to determine the first $k$ elements $v_{1:k}$ of the infinite path
it suffices to observe $x_{1:n}$ for $n$ big enough. In the present paper we show that
the barrier construction has another great advantage: namely, the process $(Z,V)=\{(Z_k,V_k)\}_{k \geq 1}$, where $V= \{V_k\}_{k \geq 1}$ denotes the Viterbi process, is under certain conditions regenerative. This is proven by, roughly speaking, applying the Markov splitting method to construct regeneration times for $Z$ which coincide with the occurrences of barriers. Regenerativity of $(Z,V)$ allows
to easily prove limit theorems to understand the asymptotic behaviour of inferences based on Viterbi
paths. In fact, in a special case of HMM this regenerative property has already been known to hold and has found several applications \cite{AV,AVacta,Vsmoothing,Vrisk,iowa}.

The paper is organized as follows. In Subsection \ref{PMC-model}, we
introduce our model and some necessary notation. In Subsection
\ref{sec:VP}, the segmentation problem, infinite Viterbi path,
barriers and many other concepts are introduced and defined. In Subsection \ref{VP} we state the two barrier construction theorems from \cite{PMMinf} (Theorem \ref{HMMTh} concerning HMM and Theorem \ref{thLSC} concerning general PMM) as well as introduce some general Markov chain terminology. In Section \ref{sec:reg} we prove our main regeneration-theorem. In Section \ref{sec:ex} we apply this theorem to more specific cases, namely HMM (Subsection
\ref{subsec:HMM}), discrete ${\cal X}$ (Subsection \ref{subsec:disc}) and linear
Markov switching model (Subsection \ref{subsec:linear}). In Section \ref{sec: VT} we demonstrate an application of our theory by providing some regeneration-based analysis of the \textit{Viterbi training} algorithm.

\subsection{Pairwise Markov model} \label{PMC-model}
Let the observation-space $\mathcal{X}$ be a
separable metric space equipped with its Borel $\sigma$-field
$\B(\X)$. Let the state-space be $\mathcal{Y}=\{1,2,\ldots,|\mathcal{Y}|\}$, where $|\Y|$ is some
positive integer. We denote $\mathcal{Z}= \X \times \Y$, and equip
$\Z$ with product topology $\tau \times 2^\Y$, where $\tau$ denotes
the topology induced by the metrics of $\X$. Furthermore, $\Z$ is
equipped with its Borel $\sigma$-field $\B(\Z)=\B(\X) \otimes 2^\Y$,
which is the smallest $\sigma$-field containing sets of the form $A
\times B$, where $A \in \B(\X)$ and $B \in 2^\Y$. Let $\mu$ be a
$\sigma$-finite measure on $\B(\X)$ and let $c$ be the counting
measure on $2^\mathcal{Y}$. Finally, let
\begin{align*}
q \colon \Z^2 \rightarrow \mathbb{R}_{\geq 0}, \quad (z,z') \mapsto q(z|z')
\end{align*}
be a such a measurable non-negative function that for each $z' \in \mathcal{Z}$ the function $z \mapsto q(z|z')$ is a density with respect to product measure $\mu \times c$.

We define random process $Z =\{Z_k \}_{k\geq 1} = \{(X_k,Y_k) \}_{k \geq 1}$ as a homogeneous Markov chain on the two-dimensional space $\mathcal{Z}$ having the transition kernel $P(z,A)$ defined by
\begin{align*}
&P(z',A)= \int_A q(z|z') \, \mu \times c (dz), \quad z' \in \mathcal{Z}, \quad A \in \B(\Z).
\end{align*}
In other words, $q(z|z')$ is the transition kernel density of $Z$. The marginal processes $\{X_k \}_{k \geq 1}$ and $\{Y_k \}_{k \geq
1}$ will be denoted with $X$ and $Y$, respectively. Following
\cite{pairwise,pairwise2,pairwise3}, we call the process $Z$ a
\textit{pairwise Markov model} (PMM). It should be noted that even though $Z$ is a Markov chain, this doesn't necessarily imply that either of the marginal processes $X$ and $Y$ are Markov chains.

The letter $p$ will be used to denote the various joint and conditional densities. By abuse of notation, the corresponding probability law is
indicated by arguments of $p(\cdot)$, with lower-case $x_k$, $y_k$ and $z_k$ indicating random variables $X_k$, $Y_k$ and $Z_k$, respectively. For example
\begin{align*}
p(x_{2:n}, y_{2:n}|x_1,y_1)= \prod_{k=2}^n q(x_k, y_k|x_{k-1},y_{k-1}),
\end{align*}
where $x_{2:n}=(x_2, \ldots,x_n)$ and $y_{2:n}=(y_2, \ldots,y_n)$. Sometimes it is convenient to use other symbols beside $x_k,y_k,z_k$ as the arguments
of some density; in that case we indicate the corresponding probability law using the equality sign, for example
\begin{align*}
p(x_{2:n}, y_{2:n}|x_1=x,y_1=i)=q(x_2,y_2|x,i)\prod_{k=3}^n q(x_k, y_k|x_{k-1},y_{k-1}), \quad n \geq 3.
\end{align*}
Also $p(z_1)=p(x_1,y_1)$ denotes the initial distribution density of $Z$ with respect to measure $\mu_1 \times c$, where $\mu_1$ is some $\sigma$-finite measure on $\B(\X)$. Thus the joint density of $Z_{1:n}$ is
$p(z_{1:n})=p(z_1)p(z_{2:n}|z_1)$. For every $n \geq 2$ and $i,j \in \Y$ we also denote
\begin{align} \label{pij}
p_{ij}(x_{1:n})=\max_{y_{1:n} \colon y_1=i, y_n=j} \prod_{k=2}^n q(x_k,y_k|x_{k-1},y_{k-1}), \quad x_{1:n} \in \X^n.
\end{align}
Thus
\begin{align*}
p_{ij}(x_{1:n}) =\max_{y_{1:n} \colon y_1=i, y_n=j}p(x_{2:n},y_{2:n}|x_1,y_1).
\end{align*}

If $p(y_2|x_1,y_1)$ doesn't depend on $x_1$, and
$p(x_2|y_2,x_1,y_1)$ doesn't depend on neither $x_1$ nor $y_1$, then
$Z$ is called a \textit{hidden Markov model} (HMM). In that case, denoting
\begin{align*}
&p_{ij}=p(y_2=j|y_1=i),\quad f_{j}(x)=p(x_2=x|y_2=j),
\end{align*}
the transition kernel density factorizes into
\begin{align*}
q(x, j|x',i)&=p(x_2=x|y_2=j,x_1=x',y_1=i)p(y_2=j|x_1=x',y_1=i)\\
&=p_{ij}f_{j}(x).
\end{align*}
Density functions $f_{j}$ are also called the \textit{emission densities}.

If $p(y_2|x_1,y_1)$ doesn't depend on $x_1$, and $p(x_2|y_2,x_1,y_1)$ doesn't depend on $y_1$, then following \cite{HMMbook} we call
$Z$ a \textit{Markov switching model}. Thus HMM's constitute a sub-class of Markov switching models. In the case of Markov switching model, denoting
\begin{align*}
&f_{j}(x|x')=p(x_2=x|y_2=j,x_1=x'),
\end{align*}
the transition kernel density becomes
\begin{align*}
q(x, j|x',i)&=p_{ij}f_{j}(x|x').
\end{align*}
It is easy to confirm that in case of Markov switching model (and therefore also in case of HMM) $Y$ is a homogeneous Markov chain with transition matrix $(p_{ij})$ \cite{pairwise}. Most PMM's
used in practice fall into the class of Markov switching models.
Figure \ref{fig:PMC} depicts the directed dependence graph of HMM,
Markov switching model and the general PMM.
\begin{figure}
    \centering
    \begin{subfigure}[b]{0.45\textwidth}
    \resizebox{\linewidth}{!}{
   \includegraphics[width=\textwidth]{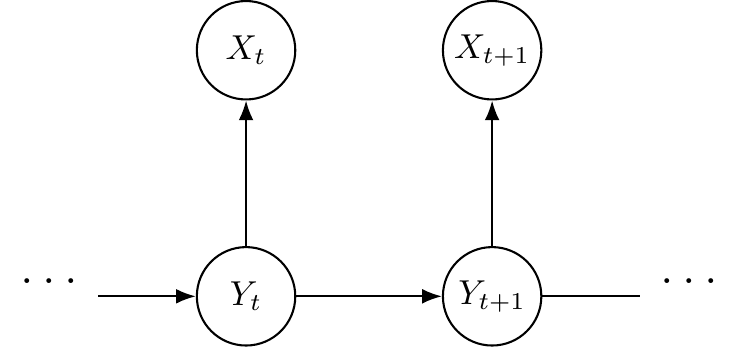}
}
        \caption{Hidden Markov model (HMM)}
    \end{subfigure}
    \quad  
      \begin{subfigure}[b]{0.45\textwidth}
  \resizebox{\linewidth}{!}{ 
  \includegraphics[width=\textwidth]{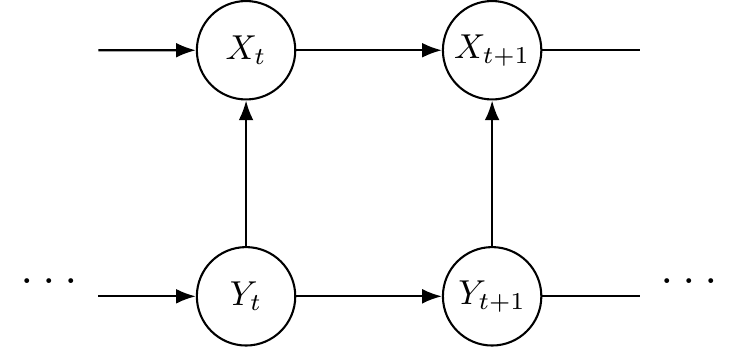}
 }
        \caption{Markov switching model}
    \end{subfigure}

\bigskip
      \begin{subfigure}[b]{0.45\textwidth}
  \resizebox{\linewidth}{!}{  
   \includegraphics[width=\textwidth]{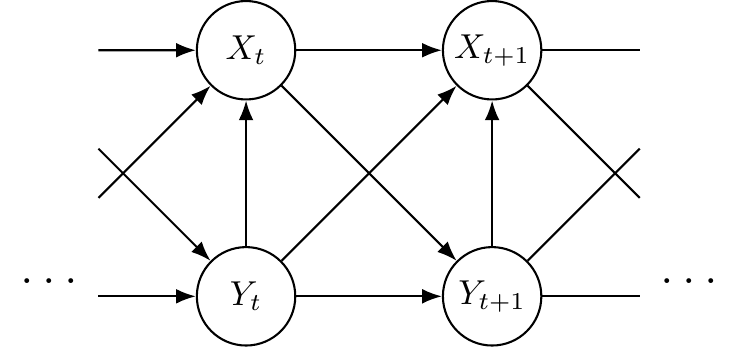}}       
    \caption{Pairwise Markov model (PMM)}
    \end{subfigure}
        \caption{Directed dependence graphs of different types of PMM's}\label{fig:PMC}
\end{figure}

Throughout this paper we write $P_z(A)=P(A|Z_1=z)$ for any event $A$ and $\E_z[S]=\E[S|Z_1=z]$ for any random variable $S$.
\subsection{Infinite Viterbi path}\label{sec:VP} The {\it segmentation problem} in general
consists of guessing or estimating the unobserved realization of
process $Y_{1:n}$ -- the true path  -- given the realization
$x_{1:n}$ of the observation process $X_{1:n}$. Since the true path
cannot be exactly known,  the segmentation procedure merely consists
of finding the path that in some sense is the best approximation.
Probably the most popular estimate is the path with maximum
posterior probability. This path will be denoted with $v(x_{1:n})$
and also with $v_{1:n}$, when $x_{1:n}$ is assumed to be fixed:
$$v_{1:n}=v(x_{1:n})=\argmax_{y_{1:n}}p(y_{1:n},x_{1:n})=\argmax_{y_{1:n}}P(Y_{1:n}=y_{1:n}|X_{1:n}=x_{1:n}).$$
Typically $v_{1:n}$ is called Viterbi or MAP path (also Viterbi or
MAP alignment). It can be found by a dynamic programming algorithm called \textit{Viterbi algorithm} \cite{EM4}. Clearly $v_{1:n}$ might not be unique. For more detailed discussion about the
segmentation problem and the properties of different estimates, we
refer to \cite{intech,seg,chris,Vrisk,peep}. Although these papers
deal with HMM's only, the general theory applies for any model
including PMM's.

To study the
statistical properties of Viterbi path-based inferences, one has to
know the long-run or typical behaviour of random vectors
$v(X_{1:n})$. As argued in \cite{AV}, behaviour of $v(X_{1:n})$ is
not trivial since the $(n+1)^{\mathrm{th}}$ observation can in
principle change the entire alignment based on the previous observations $x_{1:n}$. It might
happen with a positive probability that the first $n$ entries of
$v(x_{1:n+1})$ are all different from corresponding entries of
$v(x_{1:n})$. If this happens again and again, then the first
element of $v(x_{1:n})$ keeps changing as $n$ grows and there is not
such thing as limiting Viterbi path. On the other hand, it could be
intuitively claimed that sometimes there is a positive probability to observe
$x_{1:n}$
 such that regardless of the value of the $(n+1)^{\mathrm{th}}$ observation (provided $n$ is sufficiently large),
the paths $v(x_{1:n})$ and $v(x_{1:n+1})$ agree on first $u$
elements, where $u<n$. If this is true, then no matter what happens
in the future, the first $u$ elements of the paths remain constant.
Provided there  is an increasing unbounded sequence $u_i$
($u<u_1<u_2<\ldots$) such that the path up to $u_i$ remains
constant, one can define limiting or infinite Viterbi path. Let us
formalize the idea. In the following definition $v(x_{1:n})$ is a
Viterbi path and $v(x_{1:n})_{1:t}$ are the first $t$ elements of
the $n$-elemental vector $v(x_{1:n})$.

\begin{definition} Let $x_{1:\infty}$ be a realization of $X$. The sequence $v_{1:\infty}\in {\cal Y}^{\infty}$ is called \emph{infinite Viterbi path of $x_{1:\infty}$} if for any
$t \geq 1$ there exists $m(t)\geq t$ such that
\begin{equation}\label{infty}
v(x_{1:n})_{1:t}=v_{1:t},\quad \forall n\geq m(t).\end{equation}
\end{definition}

Hence $v_{1:\infty}$ is the infinite Viterbi path of
$x_{1:\infty}$ if for any $t$, the first $t$ elements of
$v_{1:\infty}$ are the first $t$ elements of a Viterbi path
$v(x_{1:n})$ for all $n$ big enough $(n\geq m(t))$. In other words,
for every $n$ big enough, there exists at least one  Viterbi path so
that $v(x_{1:n})_{1:t}=v_{1:t}$. For the definition of infinite Viterbi path in the case ${\cal Y}$ is infinite, see
\cite{Ritov,PMMinf}.

As we shall see, for many PMM's the infinite
Viterbi path exists for almost every realization of $X$. However, it is possible to construct a model where almost no realization has an infinite Viterbi path (see \cite[Example 1.1]{PMMinf}).

\paragraph{Nodes} Suppose now $x_{1:\infty}$ is such that infinite
Viterbi path exists. It means that for every time $t$, there exists
time $m(t)\geq t$ such that the first $t$ elements of $v(x_{1:n})$
are fixed as soon as $n\geq m$. Note that if $m(t)$ is such a time,
then $m(t)+1$ is such a time too. Theoretically, the time $m$ might
depend on the whole sequence $x_{1:\infty}$. This means that after
observing the sequence $x_{1:m}$, it is not yet clear, whether the
first $t$ elements of Viterbi path are now fixed (for any
continuation of $x_{1:m}$) or not. In practice, one would not like
to wait infinitely long, instead one prefers to realize that the
time $m(t)$ is arrived right after observing $x_{1:m}$. In this
case, the (random) time $m(t)$ is the stopping time with respect to
the observation process. In particular, it means the following: for
every possible continuation $x_{m+1:n}$ of $x_{1:m}$, the Viterbi
path at time $t$ passes the state $v_t$, let that state be $i$. This
requirement is fulfilled, when the following holds: for every two states
$j,k\in {\cal Y}$
\begin{equation}\label{r-node}
\delta_t(i)p_{ij}(x_{t:m})\geq \delta_t(k)p_{kj}(x_{t:m}),
\end{equation}
where
\begin{align*}
\delta_1(y)=p(x_1,y_1=y), \quad \delta_t(y)=\max_{y_{1:t} \colon
y_t=y}p(x_{1:t},y_{1:t}), \quad t\geq 2.
\end{align*}
and $p_{ij}(\cdot)$ is defined in (\ref{pij}). Indeed,
 there might be several states satisfying
(\ref{r-node}), but the ties can always be broken in favour of the
state $i$, so that whenever $n \geq m$, there is at least one
Viterbi path $v(x_{1:n})$ that passes the state $i$ at time $t$.
 Therefore, if at time $t$, there is a state  $i$ satisfying (\ref{r-node}), then $m$ is the time
$m(t)$ required in (\ref{infty}) and it depends on $x_{1:m}$ only.
\begin{definition} Let $x_{1:m}$ be a vector of observations.  If inequalities (\ref{r-node}) hold for any pair of states $j$ and $k$, then the
time  $t$  is called  an {\rm $i$-node of order $r=m-t$.}
Time $t$ is called a {\rm strong $i$-node of order $r$}, if it is an $i$-node of order $r$, and
the inequality (\ref{r-node}) is strict for any $j$ and $k\ne i$ for which the left side of the inequality is positive. We
call $t$ a node of order $r$ if for some $i$, it is an $i$-node of
order $r=m-t$.
\end{definition}

In \cite{OnLine} the concept of node is used under the name ``coalescence point'' to describe a memory-efficient modification of the Viterbi algorithm. Our theory justifies the use of such algorithm and could be used to derive bounds on its memory complexity.

\paragraph{Barriers}
Whether a time $t$ is a node of order $r$ or not depends, in general, on the
sequence $x_{1:t+r}$. Sometimes, however, there is some small
block of observations that guarantees the
existence of a node regardless of the other observations.  Let
us illustrate this by an example.

\begin{example}(\cite[Example 3]{PMMinf}) Suppose that   there exists a state $i\in {\cal Y}$
such that for any triplet $y_{t-1},y_t,y_{t+1} \in {\cal Y}$
\begin{equation}\label{tr-barrier}
q(x_t,i|x_{t-1},y_{t-1})q(x_{t+1},y_{t+1}|x_t,i)\geq q(x_t,y_t|x_{t-1},y_{t-1})q(x_{t+1},y_{t+1}|x_t,y_t).
\end{equation}
Then
\begin{align*}
\delta_t(i)q(x_{t+1},y_{t+1}|x_t,i)&=\max_{y'}\delta_{t-1}(y')q(x_t,i|x_{t-1},y')q(x_{t+1},y_{t+1}|x_t,i)\\
&\geq
\max_{y'}\delta_{t-1}(y')q(x_t,y_t|x_{t-1},y')q(x_{t+1},y_{t+1}|x_t,y_t)\\
&=\delta_t(y_t)q(x_{t+1},y_{t+1}|x_t,y_t).\end{align*}
We thus have that $t$ is an
$i$-node of order 1, because for every pair $j,k \in \Y$
$$\delta_t(i)p_{ij}(x_t,x_{t+1})\geq
\delta_t(k)p_{kj}(x_t,x_{t+1}).$$ Whether (\ref{tr-barrier}) holds
or not, depends on triplet $(x_{t-1},x_t,x_{t+1})$. In case of Markov switching model, (\ref{tr-barrier}) is
\begin{align*}
p_{y_{t-1}i}f_{i}(x_t|x_{t-1}) \cdot p_{i y_{t+1}} f_{y_{t+1}}(x_{t+1}|x_t) \geq 
p_{y_{t-1}y_t}f_{y_t}(x_t|x_{t-1}) \cdot p_{y_t y_{t+1}} f_{y_{t+1}}(x_{t+1}|x_t).
\end{align*}
And in a more special case of HMM, (\ref{tr-barrier}) is equivalent
to
\begin{align}\label{1-barHMM}
p_{y_{t-1}i}f_i(x_t) \cdot p_{i y_{t+1}} \geq p_{y_{t-1}y_t}f_{y_t}(x_t) \cdot p_{y_t y_{t+1}} .
\end{align}
\end{example}
The inequalities (\ref{1-barHMM}) have very clear meaning -- when
the observation $x_t$ has relatively big probability of  being
emitted from state $i$ (in comparison of being emitted from any
other state),  then regardless of the observations before or after
$x_t$, time $t$ is  a node.  On the other hand, for many models,
there are no such $x_t$ possible, so (\ref{1-barHMM}) is rather an
exception than a rule.

\begin{definition} Given $i\in \Y$, $b_{1:M}$  is called an
{\rm (strong) $i$-barrier of order $r$ and length $M$}, if, for any
$x_{1:\infty}$ with  $x_{m-M+1:m}=b_{1:M}$   for some $m\geq M$,
${m-r}$  is an (strong) $i$-node of order $r$.
\end{definition}

Suppose now that an $i$-barrier of order $r$ occurs in $x_{1:\infty}$ infinitely often. By definition of a barrier, there must exist infinitely many $i$-nodes of
order $r$ in $x_{1:\infty}$, let these nodes be $u_1<u_2<\cdots$. Let
$m \geq u_2+r$.  There must exist a Viterbi path
$v(x_{1:m})$ passing state $i$ at time $u_1$. There also exists a
Viterbi path passing $i$ at time $u_2$. If $v(x_{1:m})$ is unique,
then the path passes $i$ at both times, but if Viterbi path is
not unique and  $u_1$ and $u_2$ are too close to each other, then
there might not be possible to break ties in favour of $i$ at $u_1$
and $u_2$ simultaneously (see \cite[Example 1.4]{PMMinf}). This problem does not occur, if the nodes are strong or if  $u_2 \geq u_1+r$. Indeed, since $u_1$ is an $i$-node of order $r$, then
by definition of $r$-order node, between times $u_1$ and $u_1+r+1$,
the ties can be broken so that whatever state the Viterbi path
passes at time $u_2$, it passes $i$ at time $u_1$. Thus, if the
locations of nodes  $u_1<u_2<\cdots$ are such that $u_k\geq
u_{k-1}+r$ for all $k\geq 2$, it is possible to construct the infinite Viterbi path so that it passes the state $i$ at every time $u_k$. In what follows, when the nodes $u_k$ and $u_{k-1}$ are such
that $u_k\geq u_{k-1}+r$, then the nodes are called
\textit{separated}. Of course, there is no loss of generality in
assuming that the nodes $u_1<u_2<\cdots$ are separated, because from
any non-separated sequence of nodes it is possible to pick a separated
subsequence. Another approach is to enlarge the
barriers so that two barriers cannot overlap and, therefore, are
separated. This is the way barriers are defined in \cite{AVT5}.

\paragraph{Piecewise construction of infinite Viterbi path}
Having infinitely many separated nodes $u_1<u_2<\cdots$ or
order $r$, it is possible to construct the infinite Viterbi path
{\it piecewise}. Indeed, we know that for every $n\geq u_k+r$, there is
a Viterbi path $v_{1:n}=v(x_{1:n})$ such that $v_{u_j}=i$, $j=1,\ldots,k$.
Because of that property  and by optimality principle  clearly the
piece $v_{u_{j-1}:u_j}$ depends on the observations
$x_{u_{j-1}:u_j}$, only. Therefore $v_{1:\infty}$ can be constructed in the
following way: first use the observations $x_{1:u_1}$ to find the
first piece $v_{1:u_1}$ as follows:
$$v_{1:u_1}=\argmax_{y_{1:u_1} \colon y_{u_1}=i}p(x_{1:u_1},y_{1:u_1}).$$
Then use  $x_{u_1:u_2}$ to find the second piece $v_{u_1:u_2}$ as
follows:
$$v_{u_1:u_2}=\argmax_{y_{u_{1}:u_2} \colon y_{u_1}=y_{u_2}=i }p(x_{u_1:u_2},y_{u_1:u_2}),$$
and so on. Finally use $x_{u_k:n}$ to find the last piece
$v_{u_k:n}$ as follows:
$$v_{u_k:n}=\argmax_{y_{u_{k}:n} \colon y_{u_k}=i}p(x_{u_k:n},y_{u_{k}:n}).$$
The last piece $v_{u_k:n}$ might change as $n$ grows, but the rest
of the Viterbi path is now fixed. Thus, if $x_{1:\infty}$ contains
infinitely many nodes, the whole infinite path
can be constructed piecewise.

If the nodes $u_k$ are strong (not necessarily
separated) then regardless of tie-breaking scheme $v(x_{1:n})_{u_k}=i$
for all $k \geq 1$ and $n \geq u_k+r$. Therefore the piecewise construction detailed above is in that case
achieved when the Viterbi estimation is done by a
lexicographic\footnote{Here the term ``lexicographic'' includes both left-to-right and right-to-left lexicographic ordering (the latter is sometimes called co-lexicographic ordering).} tie-breaking scheme induced by
some ordering on $\Y$. Furthermore, under lexicographic ordering, each piece in the piecewise construction is also found by lexicographic ordering. The latter observation will be crucial to our proof of regenerativity.

If the nodes $u_k$ are not strong, then the lexicographic ordering may fail to produce the piecewise construction. However from practice point of view the lexicographic ordering is clearly advantageous, since it can easily be implemented via Viterbi algorithm. For that reason we will from now on almost exclusively focus on strong nodes and strong barriers. Fortunately, this does not seem to impose any significant restrictions.
\subsection{Viterbi process} \label{VP} The notion of infinite Viterbi path of a fixed realization $x_{1:\infty}$ naturally carries over to an infinite Viterbi path of $X$,
called the \textit{Viterbi process}. Formally,
this process is defined as follows.
\begin{definition} A random process $V=\{V_{k}\}_{k \geq 1}$ on space $\Y$ is called a \emph{Viterbi process}, if the event $\{V \mbox{ is not an infinite Viterbi path of }X \}$
is contained in a set of zero probability measure.
\end{definition}
If there exists a barrier set $\X^* \subset \X^M$ consisting of $i$-barriers of fixed
order and satisfying $P(X \in \X^* \mbox{ i.o.})=1$, where
\begin{align*}
\{ X \in \mathcal{X}^* \mbox{ i.o.} \} \DEF \bigcap_{k=1}^\infty \bigcup_{l=k}^\infty \{X_{l:l+M-1} \in \mathcal{X}^* \} ,
\end{align*}
then the
Viterbi process can be constructed by applying the piecewise
construction detailed above to the process $X$. Having infinitely many barriers not only ensures the existence of the Viterbi process, but -- as we shall see briefly -- will also provide a rather straightforward route to proving that the Viterbi process is regenerative. As follows we state two theorems which ensure the existence of Viterbi process. The first one only concerns HMM while the latter applies to any PMM.

Recall that in case of HMM $f_i$ are the emission densities with respect to measure $\mu$. Denote
\begin{align*}
G_i=\{x\in \X \:|\:  f_i(x)>0\}, \quad i \in \Y.
\end{align*}
The conditions for the HMM-theorem are the following.
\begin{description}
\item[(HMM1)] \label{HMMcon1} For each state $j\in \Y$
\begin{align*}
&\mu\left(\left\{x\in\mathcal{X} \: | \: f_j(x)p_{ \Cdot j}> \max_{i\in \Y,~i\ne
j}f_i(x)p_{\Cdot i}\right\}\right)>0,\quad\text{where}~p_{\Cdot j} \DEF
\max_{i\in \Y}p_{ij}.
\end{align*}
\item[(HMM2)] There exists a set $C\subset\mathcal{Y}$ such that
\begin{equation*}
\mu \left[ \left(\cap _{i\in C}G_i \right)\setminus \left(\cup _{i\notin C}G_i \right) \right]>0
\end{equation*}
and the sub-stochastic matrix
$\mathbb{P}_C=(p_{ij})_{i,j\in C}$ is primitive in the sense
that $\mathbb{P}^R_C$ consists of only positive elements for some positive
integer $R$.
\end{description}

Note that condition (HMM2) only depends on transition matrix $\mathbb{P} = (p_{ij})$ and the probability laws induced by densities $f_i$. Also, this condition is not very restrictive. It is fulfilled when $\mathbb{P}$ consists of only positive elements. It is also fulfilled when $\mathbb{P}$ is primitive and there exists some $\mu$-positive set such that $f_i$ are positive on that set.
\begin{theorem} \cite[Cor. 4.1]{PMMinf} \label{HMMTh} Let $Z$ be HMM satisfying (HMM1) and (HMM2) and let Markov chain $Y$ be irreducible. There exists a set $\X^* \subset \X^M$ consisting of strong $i$-barriers of fixed order and satisfying $P(X \in \X^* \mbox{ i.o.})=1$.
\end{theorem}
Theorem \ref{HMMcon1} is not the only result in literature guaranteeing the existence of Viterbi process for HMM, but it is the latest and has the most general conditions (cf. \cite{caliebe1,caliebe2,AVT5,K2}).

In order to state the PMM-theorem we need introduce some general state Markov chain terminology.  Markov chain $Z$ is called
\textit{$\varphi$-irreducible} for some $\sigma$-finite measure
$\varphi$ on $\B(\Z)$, if $\varphi(A)>0$ implies $\sum_{k=2}^\infty
P_z(Z_k \in A)>0$ for all $z \in \Z$. If $Z$ is
$\varphi$-irreducible, then there exists (see \cite[Prop.
4.2.2.]{MT}) a \textit{maximal irreducibility measure} $\psi$ in the
sense that for any other irreducibility measure $\varphi'$ the
measure $\psi$ dominates $\varphi'$, $\psi \succ \varphi'$. The
symbol $\psi$ will be reserved to denote the maximal irreducibility
measure of $Z$.  A point $z\in \Z$ is called \textit{reachable} if for
every open neighbourhood $O$ of $z$,
\begin{align*}
\sum_{k =2}^\infty P_{z'}(Z_k \in O)>0, \quad \forall z' \in \Z.
\end{align*}
For $\psi$-irreducible $Z$, the point $z$ is reachable if and only
if it belongs to the support of $\psi$ \cite[Lemma 6.1.4]{MT}. Since we have equipped space $\Z$ with product topology $\tau \times 2^\Y$, where $\tau$ denotes the topology induced by the metrics of
$\X$, the above-stated definition of reachable point is in fact
equivalent to the following: point $(x,i) \in \Z$ is called
\emph{reachable}, if for every open neighbourhood $O$ of $x$,
\begin{align*}
\sum_{k =2}^\infty P_z(Z_k \in O \times \{i\})>0, \quad \forall
z \in \Z.
\end{align*}
Chain $Z$ is called \textit{Harris recurrent}, if it is
$\psi$-irreducible and $\psi(A)>0$ implies $P_z(Z_k \in A \mbox{
i.o.})=1$ for all $z \in \Z$.

For any $n \geq 2$ define
\begin{align} \label{Yplus}
\mathcal{Y}^+(x)= \{ (i,j) \: | \: p_{ij}(x)>0\}, \quad x \in \X^n.
\end{align}
For any set $A$ consisting of vectors of length $n$ we adopt the following notation:
\begin{align*}
&A_{(k)}=\{x_{k} \:| \: x_{1:n} \in A\}, \quad 1 \leq k \leq n,\\
&A_{(k,l)}=\{x_{k:l} \:| \: x_{1:n} \in A\}, \quad 1 \leq k \leq l \leq n.
\end{align*}
Hence
\begin{align*}
&{\Y}^+(x)_{(1)}=\{i \: | \: \exists j(i) \text{   such that }
p_{ij}(x)>0\},\\
& \mathcal{Y}^+(x)_{(2)}=\{j \: | \: \exists i(j)
\text{  such that } p_{ij}(x)>0\}.
\end{align*} Observe that if $i\in
\mathcal{Y}^+(x)_{(1)}$ and $j\in \mathcal{Y}^+(x)_{(2)}$ then not
necessarily $(i,j)\in \mathcal{Y}^+(x).$

The conditions for the PMM-theorem are the following.
\begin{description} \item[(PMM1)] For any $N \geq 1$ there exist $1<n_{1}< \ldots <n_{2N}$, an open set $B \subset \X^{n_{2N}}$, and $\epsilon>0$ such that denoting $n_0=1$ we have for all $k=1, \ldots, 2N$ and all $x \in B_{(n_{k-1},n_k)}$
\begin{align}
&p_{11}(x) > p_{i1}(x),\quad \forall i \in \mathcal{Y}, \label{PMMineq1}\\
&p_{11}(x) > p_{1i}(x) ,\quad \forall i \in \mathcal{Y}, \label{PMMineq2} \\
& p_{11}(x)(1-\epsilon) > p_{ij}(x), \quad \forall i,j \in \mathcal{Y} \setminus \{1\}, \notag
\end{align}
where either inequalities \eqref{PMMineq1} or \eqref{PMMineq2} could be non-strict. We assume that $\epsilon$ and $B_{(1)}$ are independent of $N$, and that there exists a compact set $K$, which is independent of $N$, such that $B_{(n_{2N})}$ is contained in $K$. Furthermore, we assume that there exists $x^* \in B_{(1)}$ such that $(x^*,1)$ is reachable.
\item[(PMM2)]There exists an open set $E \subset\X^q$, $q\geq 2$, such that $\Y^+ \stackrel{\scriptsize{\mbox{def}}}{=}\Y^+(x)$
 is the same for every $x \in E$ and satisfies the following property: $(i,j) \in \Y^+$ for every $i \in \Y^+_{(1)}$ and $j \in \Y^+_{(2)}$. Furthermore, we assume
 that there exists a reachable point in $E_{(1)}\times \Y^+_{(1)}$.
\end{description}
\begin{theorem} {\cite[Th. 3.1]{PMMinf}} \label{thLSC} Let $\mu$ be strictly
postive\footnote{A measure is called \textit{strictly positive} if
it assigns a positive measure to all non-empty open sets.} and let
for every pair of states $i,j \in \Y$ function $(x,x') \mapsto
q(x,j|x',i)$ be lower semi-continuous and bounded. If $Z$ satisfies
(PMM1) and (PMM2), then there exists $\X^* \subset \X^M$ consisting of strong 1-barriers of fixed order. Moreover, if $Z$ is
Harris recurrent, then $P(X \in \X^* \mbox{ i.o.})=1$.
\end{theorem}

For discussion on conditions (PMM1) and (PMM2) and applications of Theorem \ref{thLSC} to specific models we refer the reader to \cite{PMMinf}.
\section{Regenerativity} \label{sec:reg}
For any sequence $u=(u_k)_{k \geq 1}$ and a sequence of times $s=(s_k)_{k \geq 1}$, $1 \leq s_1 < s_2, \ldots$, a \textit{shift operator} $\theta_t$, $t \geq 1$, is defined by $\theta_t(u,s)=((u_k)_{k \geq t}, (s_k-t+1)_{k \geq n(t)})$, where $n(t) = \min \{n \:| \: s_n \geq t  \}$. A process $\{U_k\}_{k \geq 1}$ is called \textit{regenerative} \cite{kalashnikov,thorisson} (in the classic sense), if there exists a sequence of random times $S=\{S_k\}_{k \geq 1}$, $1 \leq S_1< S_2 < \ldots$, called \textit{regeneration times}, such that for each $n \geq 1$
\begin{align*}
&\theta_{S_n}(U,S)\stackrel{d}{=} \theta_{S_1}(U,S),\\
&\theta_{S_n}(U,S) \mbox{ is independent of } (\{U_k\}_{k=1}^{S_n-1}, S_1,\ldots, S_n ).
\end{align*}
Random variables $S_{k}-S_{k-1}$ are called \textit{inter-regeneration times}. Typically one is interested in the case where inter-regeneration times have finite mean, i.e. $\mathbb{E}[S_2-S_1]< \infty$.

Let from now on $V=\{V_k\}_{k \geq 1}$ denote the Viterbi process, whenever it exists, obtained by a lexicographic tie-breaking scheme. The main goal of the present paper is to prove that under some conditions the joint process $(Z,V)=\{(Z_k,V_k)\}_{k \geq 1}$ is regenerative. Regenerativity of a general state space Markov chain is typically proved by applying the so-called \textit{splitting method} \cite{kalashnikov,thorisson}. If the regeneration times of $Z$ obtained by the splitting method are also strong nodes, then $(Z,V)$ is regenerative as well. This is the idea behind our main regeneration theorem below. The theorem is proven under assumptions (R1)-(R3), where (R1) ensures the existence of an appropriate barrier set, (R2) ensures that the splitting method can be applied to obtain suitable regeneration times for $Z$ and (R3) guarantees that the inter-regeneration times have finite expectation.

For any $A \in \B(\Z)$ let $\tau_A$ denote the number of time-steps for $Z$ to reach $A$ after time 1:
\begin{align*}
\tau_A=\min\{n \geq 1 \:|\: Z_{n+1} \in A\}.
\end{align*}
\begin{description} \item[(R1)\namedlabel{R}{(R1)}] There exists a set $\X^*\subset \X^{M}$ satisfying the following conditions:
\begin{enumerate}[label=(\roman*)] \item \label{Rbar} $M \geq 3$ and $\X^*$ consists of strong $1$-barriers of order $r \in \{1, \ldots, M-2\}$;
\item \label{Rio}there exist state $i_0 \in \Y$, such that for every $z \in \Z$
\begin{align*}
&P_z \left(Z_{1:M-r-1} \in \widetilde{\Z} \mbox{ i.o.} \right)=1,
\end{align*}
where $\widetilde{\Z} \DEF \X^*_{(1,M-r-1)} \times (\Y^{M-r-2} \times \{i_0\})$;
\item \label{Rprod} it holds $\X^*=\X^*_{(1,M-r-1)}\times \X^*_{(M-r,M)}$;
\item \label{Rpos} there exists $i_1 \in \Y$ such that for all $x \in \X^*_{(M-r)}$ $$P(X_{M-r+1:M} \in \X^*_{(M-r+1,M)}| Z_{M-r}=(x,i_1))>0.$$
\end{enumerate}
\item[(R2)\namedlabel{minor}{(R2)}] Denote $\Z_0=\X^*_{(M-r-1)} \times \{i_0\}$ and $\Z_1=\X^*_{(M-r)} \times \{i_1\}$. There exists a probability measure $\nu$ on $\mathcal{B}(\Z)$ and $\beta \in (0,1)$ such that $\nu(\Z_1)=1$ and
\begin{align*}
P(z,A) \geq \beta \nu(A),  \quad \forall z \in \Z_0, \quad \forall A \in \mathcal{B}(\Z).
\end{align*}
\item[(R3)\namedlabel{recurr}{(R3)}] It holds $\sup_{z \in \Z_0}\mathbb{E}_{z} [\tau_{\Z_0}]< \infty$.
\end{description}

\begin{theorem}\label{reg} If \ref{R} and \ref{minor} hold, then
\begin{enumerate}[label=(\roman*)] \item $\X^*$ satisfies $P(X \in \X^* \mbox{ i.o.})=1$ and hence the Viterbi process $V$ is well-defined;
\item  the process $(Z,V)=\{(Z_k,V_k)\}_{k \geq 1}$ is regenerative.
\end{enumerate}
 If also \ref{recurr} holds, then the inter-regeneration times have finite mean.
\end{theorem}
\begin{proof} \emph{Construction of regeneration times.} Denote
\begin{align*}
Q(z,A)=\dfrac{P(z, A)- \beta \nu(A)}{1-\beta}, \quad z \in \mathcal{Z}_0.
\end{align*}
By  \ref{minor} we have that $Q(z,\cdot)$ are probability measures. Hence for all $z \in \Z_0$ the measure $P(z, \cdot)$ is a mixture of measures $\nu$ and $Q(z,\cdot)$:
\begin{align} \label{splitEq}
P(z,A)=\beta \nu(A)+(1-\beta)Q(z,A), \quad z \in \Z_0.
\end{align}
We now apply the standard splitting technique to construct regeneration times for Markov chain $Z$. We define a homogeneous Markov chain $W=\{W_k\}_{k \geq 1}=\{(W_k^{1},W_k^2) \}_{k \geq 1}$ on space $\Z \times \{0,1\}$ as follows. Whenever $W^1_k \notin \Z_0$ we take $W^2_k=0$ and generate $W^1_{k+1}$ according to $P(W^1_k, \cdot)$, and whenever $W^1_k \in \Z_0$ we flip a coin: with probability $1-\beta$ we take $W^2_k=0$ and generate $W^1_{k+1}$ according to $Q(W^1_{k}, \cdot)$ and with probability $\beta$ we take $W^2_k=1$ and generate $W^1_{k+1}$ according to $\nu$ (independently of $W^1_k$).
%
The initial distribution of $W$ is defined by
\begin{align*}
&P(W_1^{1} \in A, W_1^{2}=1)=\beta P(Z_1 \in A \cap \Z_0),\\
&P(W_1^{1} \in A, W_1^{2}=0)=(1-\beta) P(Z_1 \in A \cap \Z_0)+P(Z_1 \in A \cap \Z_0^\mathsf{c}).
\end{align*}
It follows from \eqref{splitEq} that $\{W^{1}_k\}_{k \geq 1} \stackrel{d}{=}Z$. Let
\begin{align*}
&S_1^1=\min\{n+1 \: | \: W_n \in \Z_0 \times \{1\} \},\\
&S_k^1=\min\{n+1 > S_{k-1}^1\: | \: W_n \in \Z_0 \times \{1\} \}, \quad k\geq 2,
\end{align*}
Since $S_k^1-1$ are stopping times for $W$ and $W_{S_k^1-1} \in \Z_0 \times \{1\}$, we have by strong Markov property that $S_k^1$ are regeneration times for $W$ and therefore also for $Z \DEF \{W^{1}_k \}_{k \geq 1}$. In fact, for every $k \geq 1$, $\{Z_{S_k^1+u}\}_{u \geq 0}$ is a Markov chain with initial distribution $\nu$ and transition kernel $P(z, A)$.

Let $b= (r+1) \vee (M-r-1)$, where $\vee$ denotes maximum, and let
\begin{align*}
S_k^2(l)=S^1_{l+(k-1)b}, \quad k\geq 1, \quad l=1,\ldots,b.
\end{align*}
Times $S_k^2(l)$ are regeneration times for $Z$ (since $S_k^1$ are regeneration times), with $\{Z_{S^2_k(l)+u}\}_{u \geq 0}$ being a Markov chain with initial distribution $\nu$ and transition kernel $P(z, A)$ for each $k \geq 1$. But this time we always have some separation between regeneration times:
\begin{align} \label{separ}
S^2_{k}(l)-S^2_{k-1}(l) \geq b, \quad k \geq 2.
\end{align}

Define
\begin{align*}
I_{k}(l)=\mathbb{I}_{\widetilde{\Z}}(Z_{S^2_{k}(l)-M+r+1:S^2_{k}(l)-1}), \quad k \geq 2,
\end{align*}
where $\mathbb{I}$ denotes the indicator function and $\widetilde{\Z}$ is as defined in \ref{R}\ref{Rio}. Since $b \geq M-r-1$, then \eqref{separ} implies that $\{I_k(l)\}_{k \geq 2}$ are i.i.d. sequences. We will show that there exists $l^*$ such that
\begin{align} \label{Ipos}
P(I_2(l^*)=1)>0.
\end{align}
Denote
\begin{align*}
&J_1= \min \left\{n \: |  \:  Z_{n-M+r+2:n} \in \widetilde{\Z} \right\}, \\
&J_k= \min \left\{n > J_{k-1} \: |  \:  Z_{n-M+r+2:n} \in \widetilde{\Z} \right\}, \quad k \geq 2.
\end{align*}
By \ref{R}\ref{Rio} every $J_k$ is almost surely finite. Denote $\{U_k \}_{k \geq 1}= \{ W^2_{J_k} \}_{k \geq 1}$ and note that $\{U_k\}_{k \geq 1}$ is i.i.d. sequence with $P(U_1=1)=\beta$. Indeed, fix $k \geq 1$ and consider the homogeneous Markov chain $((Z_1,W^2_0),(Z_2,W^2_1),(Z_3,W^2_2), \ldots)$, where $W^2_0 \DEF 0$. Since $J_k$ is stopping time for that Markov chain, we have by strong Markov property for every $z \in \Z_0$
\begin{align*}
P(U_k=1|(Z_{J_k},W^2_{J_k-1})=(z,w))=P(W^2_{J_k}=1|(Z_{J_k},W^2_{J_k-1})=(z,w))=\beta.
\end{align*}
Since $Z_{J_k} \in \Z_0$ by definition of $J_k$, then $P(U_k=1)=\beta$. Furthermore, by strong Markov property $U_k$ depends on $U_{1:k-1}$ only through $(Z_{J_k},W^2_{J_k-1})$. But, as we saw, $U_k$ is independent of $(Z_{J_k},W^2_{J_k-1})$ and therefore also of $U_{1:k-1}$. This implies that $\{U_k \}_{k \geq 1}$ is i.i.d. as claimed. Thus
\begin{align*}
\sum_{l=1}^{b}P \left( \sum_{k=1}^\infty I_k(l)=\infty \right)& \geq P \left( \bigcup_{l=1}^{b} \left\{ \sum_{k=2}^\infty I_{k}(l)=\infty  \right\} \right)\\
& = P \left( \sum_{k=M}^\infty \mathbb{I}_{\widetilde{\Z}}(Z_{k-M+r+2:k}) \cdot W^2_k =\infty \right)\\
& = P \left(\sum_{k=1}^\infty U_k=\infty  \right)\\
&=1,
\end{align*}
which implies that $l^*$ satisfying \eqref{Ipos} indeed exists.

Take now
\begin{align*}
&S^3_1=\min\{S^2_n(l^*) \: | \: I_{n}(l^*)=1, n\geq 2 \},\\
&S^3_k=\min\{S^2_n(l^*) >S^3_{k-1} \: | \: I_{n}(l^*)=1, n\geq 2 \}, \quad k\geq 2.
\end{align*}
Note that by \eqref{Ipos} $S_k^3$ are almost surely finite. Again, $S^3_k-1$ are stopping times for $W$ and $W_{S^3_{k}-1} \in \Z_0 \times \{1\}$, so it follows from strong Markov property and \eqref{separ} that $S^3_k$ are regeneration times for $Z$, with $\{Z_{S^3_k+u}\}_{u \geq 0}$ being a Markov chain with initial distribution $\nu$ and transition kernel $P(z, A)$ for each $k \geq 1$. Also, by construction $S^3_k$ satisfies $X_{S^3_k-M+r+1:S^3_k-1} \in \X^*_{(1:M-r-1)}$ for $k \geq 1$.

Finally, we take
\begin{align*}
&S_1=\min \{S^3_n \:| \: X_{S^3_n:S^3_n+r} \in \X^*_{(M-r,M)},n\geq 1\},\\
&S_k=\min \{S^3_n> S_{k-1} \:| \: X_{S^3_n:S^3_n+r} \in \X^*_{(M-r,M)},n\geq 1 \}, \quad k\geq 2.
\end{align*}
Note that by \ref{R}\ref{Rpos} and the assumption $\nu(\Z_1)=1$ of \ref{minor} $S_k$ are almost surely finite. We will now show $S_k$ are regeneration times for $W$. Let
\begin{align*}
H_k=\mathbb{I}_{\X^*_{(M-r,M)}}(X_{S^3_k:S^3_k+r}), \quad k \geq 1,
\end{align*}
and let $N(k)$ be such that $S^3_{N(k)}=S_k$. Hence $N(k)=n$ if and only if $\sum_{l=1}^{n-1} H_l =k-1$ and $H_n=1$. Also note that by \eqref{separ} $S_{k}-S_{k-1} \geq r+1$. We can express the conditional distribution of $W_{S_k:S_k+u}$, $u \geq 1$, given past history:
\begin{align*}
P(W_{S_k:S_k+u} \in A |W_{1:S_k-1}, S_{1:k})&
 =P(W_{S_k:S_k+u} \in A |W_{1:S_k-1})\\
&  =\sum_{n=1}^\infty P(W_{S^3_n:S^3_n+u} \in A |W_{1:S^3_n-1},H_{n}) \cdot \mathbb{I}(N(k)=n) \\
&  =\sum_{n=1}^\infty P(W_{S^3_n:S^3_n+u} \in A |H_{n}) \cdot \mathbb{I}(N(k)=n) \\
&  =\sum_{n=1}^\infty P(W_{S^3_n:S^3_n+u} \in A |H_{n}=1) \cdot \mathbb{I}(N(k)=n) \\
&  =P(W_{S^3_1:S^3_1+u}\in A|H_1=1).
\end{align*}
Here the first equality follows because $\sigma(W_{1:S_k-1}, S_{1:k})=\sigma(W_{1:S_k-1})$; the second equality can be easily confirmed using the fact that $B \cap \{N(k)=n\} \in \sigma(W_{1:S^3_n-1},H_n)$ for every $B \in \sigma(W_{1:S_{k}-1})$ and $n \geq 1$; the third equality follows from the fact that $W_{1:S^3_n-1}$ is independent of $(W_{S^3_n:S^3_n+u},H_n)$ (recall that $S^3_n$ is a regeneration time); the fourth equality follows from the fact that $\{N(k)=n\} \subset \{H_n=1\}$. That $S_k$ are regeneration times for $W$ easily follows. Also, we have $X_{S_k-M+r+1:S_k+r} \in \X^*$ for all $k \geq 1$ by construction and by the assumption $\X^*=\X^*_{(1,M-r-1)}\times \X^*_{(M-r,M)}$ (\ref{R}\ref{Rprod}). Therefore $P(X \in \X^* \mbox{ i.o.})=1$. Also by \ref{R}\ref{Rbar} $S_k$ are strong 1-nodes, and so by piecewise construction with lexicographic tie-breaking inside each piece the process $\{(Z_k,V_k)\}_{k \geq 1}$ is regenerative as claimed.
\bigskip

\noindent \emph{Inter-regeneration times have finite mean.} We will now show that $\mathbb{E}[S_2-S_1]<\infty$. Denote $m=\sup_{z \in \Z_0} \E_z [\tau_{\Z_0}]$, so that by \ref{recurr} $m< \infty$. We write $P_w(A)=P(A|W_1=w)$ for any event $A$ and $\E_w[S]=\E[S|W_1=w]$ for any random variable $S$. First we show that
\begin{align} \label{ineq:Mbeta}
\E_{w}[\tau_{\Z_0}] \leq \dfrac{m}{\beta \wedge (1-\beta)}, \quad \forall w \in \Z_0 \times \{0,1\},
\end{align}
where $\wedge$ denotes minimum. Indeed, note that for any $k \geq 1$ and $z \in \Z_0$
\begin{align*}
P_z(\tau_{\Z_0}=k)&=P_z(W_1^2=0)P_{(z,0)}(\tau_{\Z_0}=k)+P_{z}(W_1^2=1)P_{(z,1)}(\tau_{\Z_0}=k)\\
&=(1-\beta)P_{(z,0)}(\tau_{\Z_0}=k)+\beta P_{(z,1)}(\tau_{\Z_0}=k).
\end{align*}
Therefore, for every $z \in \Z_0$,
\begin{align*}
m &\geq \E_z[\tau_{\Z_0}]\\
&= \sum_{k=1}^\infty k \cdot [(1-\beta)P_{(z,0)}(\tau_{\Z_0}=k)+\beta P_{(z,1)}(\tau_{\Z_0}=k)]\\
&=(1-\beta)\E_{(z,0)}[\tau_{\Z_0}]+\beta \E_{(z,1)}[\tau_{\Z_0}].
\end{align*}
This implies \eqref{ineq:Mbeta}.

Next, we will show that $\mathbb{E}[S^1_2-S^1_1]< \infty$. Denote
\begin{align*}
&T_1= \min \{n \geq 1 \:| \: Z_{S_1^1-1 +n} \in \Z_0\},\\
&T_k= \min \{n > T_{k-1} \:| \: Z_{S_1^1-1 +n} \in \Z_0\}, \quad k\geq 2.
\end{align*}
Exactly like in the case of $\{U_k\}_{k \geq 1}$, it can be shown that $ \{W^2_{S_{1}^1-1+T_k}\}_{k \geq 1}$ is i.i.d. with $P(W^2_{S_{1}^1-1+T_1}=1)=\beta$. Set $N=\min \{n \geq 1 \:| \: W^2_{S_{1}^1-1+T_n} =1\}$ and $T_0=0$. Note that for any $k \geq 1$
\begin{align*}
\E[(T_k-T_{k-1})\mathbb{I}(N \geq k)]&=\E[ \E[(T_k-T_{k-1})\mathbb{I}(N \geq k)|W_{1:S_1^1-1+T_{k-1}}]]\\
&=\E[\mathbb{I}(N \geq k) \E[T_k-T_{k-1}|W_{1:S_1^1-1+T_{k-1}}]]\\
&=\E[ \mathbb{I}(N \geq k) \E_{W_{S_1^1-1+T_{k-1}}}[\tau_{\Z_0}]]\\
&\leq \E \left[ \mathbb{I}(N \geq k) \dfrac{m}{\beta \wedge (1-\beta)} \right]\\
&=\dfrac{m}{\beta \wedge (1-\beta)} \cdot P(N \geq k).
\end{align*}
Here we used the fact that $\{N \geq k\} \in \sigma(W_{1:S_1^1-1+T_{k-1}})$ for second equality, that $S_1^1-1+T_{k-1}$ is a stopping time for $W$ and strong Markov property for third equality, and \eqref{ineq:Mbeta} for the inequality. Thus we have
\begin{align*}
\mathbb{E}[S_2^1-S^1_1] &= \E \left[ \sum_{k=1}^N (T_k-T_{k-1})\right] \\
&=   \E \left[ \sum_{k=1}^\infty \sum_{n=k}^\infty (T_k-T_{k-1}) \mathbb{I}(N=n)\right]\\
&=  \sum_{k=1}^\infty \E \left[  (T_k-T_{k-1}) \mathbb{I}(N \geq k)\right] \\
& \leq \dfrac{m}{\beta \wedge (1-\beta)} \sum_{k=1}^\infty P(N \geq k) \\
&=\dfrac{m}{\beta \wedge (1-\beta)} \E[N].
\end{align*}
Since $N$ has a geometric distribution, then $c_0 \DEF \mathbb{E}[S_2^1-S^1_1]< \infty$, as claimed.

It follows that $\mathbb{E}[S^2_2(l^*)-S^2_1(l^*)]=c_0 b$. Next, let $K$ be such that $S^2_K(l^*)=S^3_2$. Then $K-3$ has negative binomial distribution,
\begin{align*}
K-3 \sim NB(2,1-P(I_2(l^*)=1)).
\end{align*} Therefore
\begin{align*}
\mathbb{E}[S_2^3-S_1^3] \leq \mathbb{E}[S^2_{K}(l^*)-S^2_1(l^*)]= \mathbb{E}[K-1] \cdot c_0 b < \infty,
\end{align*}
where the equality follows by Wald's equation.

Finally, let $L$ be such that $S^3_L=S_2$. Then $L-2$ has negative binomial distribution,
\begin{align*}
L-2 \sim NB(2,1-P(X_{S^3_1:S^3_1+r} \in \X^*_{(M-r,M)})),
\end{align*} and applying Wald's equation again gives the desired result:
\begin{align*}
\mathbb{E}[S_2-S_1] \leq \mathbb{E}[S^3_{L+1}-S^3_1]= \mathbb{E}[L] \cdot \mathbb{E}[K-1]  \cdot c_0 b< \infty.
\end{align*}
\end{proof}

Recall the sequence $n_1 < \ldots < n_N$ and set $B \subset \X^{n_N}$ of (PMM1). A set $A \in \B(\Z)$ is called \textit{regular} when $Z$ is $\psi$-irreducible, if for all $B$ satisfying $\psi(B)>0$, $\sup_{z \in A} \E_z[\tau_B]< \infty$. When conditions of Theorem \ref{thLSC} hold, then regularity of $ (B_{(n_N-1)} \times \Y) $ is sufficient for \ref{R}-\ref{recurr} to hold:
\begin{lemma}\label{PMMLem} If the following conditions are fulfilled, then \ref{R}-\ref{recurr} hold:
$\mu$ is strictly positive, function $(x',x') \mapsto q(x,j|x',i)$ is lower semi-continuous and bounded for all $i, j \in \Y$, $Z$ is Harris recurrent, (PMM1) holds with set $(B_{(n_N-1)} \times \Y) $ being regular, and (PMM2) holds.
\end{lemma}
\begin{proof} We omit some of the technical details from the proof since they largely concern the specific construction of the barrier set given in \cite{PMMinf}. By Theorem \ref{thLSC} there exists barrier set $\X^* \subset \X^M$ consisting of strong 1-barriers of order $r \in \{1, \ldots, M-2\}$. Using the lower semi-continuity assumption, this barrier set can be constructed so that it satisfies \ref{R} with some $i_0 \in \Y$ and $i_1=1$. Furthermore, denoting $B' =B_{(n_N-1,n_N)}$, we may with no loss of generality assume that $q(x,1|x',i_0)>0$ for all $(x',x) \in B'$. There exists $\epsilon_0>0$ such that set
\begin{align} \label{Bset}
\{(x',x) \in B' \:| \: q(x,1|x',i_0)>\epsilon_0 \}
\end{align}
is non-empty. Since $B'$ is open, then by the lower semi-continuity assumption \eqref{Bset} is also open, so there exist open balls $B'_0,B'_1 \subset \X$ such that $B_0' \times B_1'$ is contained in \eqref{Bset}. Thus $B_0' \times B_1' \subset B'$ and $q(x,1|x',i_0)> \epsilon_0$ for all $(x',x) \in B_0' \times B_1'$. By construction of $\X^*$ (see \cite{PMMinf} for details) we may with no loss of generality assume that $\X^*_{(M-r-1,M-r)}=B_0' \times B_1'$. For any $A \in \B(\Z)$ write $A_1$ for $\{x \:|\: (x,1) \in A\}$. Take $\beta=\epsilon_0 \mu(B_1')$ and
\begin{align*}
\nu(A)=\dfrac{\mu( A_1 \cap B'_1 )}{\mu(B'_1)}, \quad A \in \B(\Z).
\end{align*}
Then for all $x' \in B'_0$
\begin{align*}
P((x',i_0),A) &\geq P((x',i_0),A_1 \times \{1\})\\
&=\int_{A_1} q(x,1|x',i_0) \, \mu(dx)\\
& \geq \epsilon_0 \mu(A_1 \cap B'_1)\\
& =\beta \nu(A),
\end{align*}
and so \ref{minor} holds. Since $\X^*_{(M-r-1)} =B'_0 \subset B_{(n_N-1)}$ and since  \ref{R}\ref{Rio} implies $\psi(\X^*_{(M-r-1)} \times \{i_0\})>0$ (see \cite[Prop. 8.3.7]{MT}), then  \ref{recurr} follows by regularity of $B_{(n_N-1)} \times \Y$.
\end{proof}

Markov chain $Z$ is called \textit{positive}, if its transition kernel admits an invariant probability measure. It is of interest to note that for \ref{R}-\ref{recurr} to hold, $Z$ must be both Harris recurrent and positive:

\begin{lemma} \label{lem:HarrisPos} If \ref{R}-\ref{recurr} hold, then $Z$ is Harris recurrent and positive.
\end{lemma}
\begin{proof} Since \ref{R}\ref{Rio} implies that $P_z(Z_k \in \Z_0 \mbox{ i.o.})=1$ for every $z \in \Z$, then for every $z \in \Z$ there exists $k(z) \geq 2$ such that $P_z(Z_{k(z)} \in \Z_0)>0$. Also note that by \ref{minor}
\begin{align*}
P_z(Z_{z(k)+1} \in A) \geq \int_{\Z_0}P(z',A) \, P_z(Z_{z(k)} \in dz') \geq \beta \nu(A) \cdot P_z(Z_{k(z)} \in \Z_0).
\end{align*}
It follows that $Z$ must be $\nu$-irreducible. That $Z$ is Harris recurrent now follows from \cite[Prop. 9.1.7(ii)]{MT} and that $Z$ is positive follows from \ref{recurr} and \cite[Th. 11.0.1]{MT}.
\end{proof}
\section{Examples} \label{sec:ex}
\subsection{Hidden Markov model} \label{subsec:HMM}
We consider some simple examples of applications of Theorem \ref{reg} to specific models. For HMM no additional conditions than those of barrier construction theorem \ref{HMMTh} are needed for \ref{R}-\ref{recurr} to hold.

\begin{lemma} If $Z$ is HMM, $Y$ is irreducible and (HMM1)-(HMM2) hold, then \ref{R}-\ref{recurr} are satisfied.
\end{lemma}
\begin{proof} The assumptions of the statement guarantee the existence of the barrier set $\X^* \subset \X^M$ consisting of strong $1$-barriers with fixed order $r$ for some labelling of $\Y$. Furthermore, it follows from the construction of $\X^*$ that \ref{R} holds with $i_1=1$ and $i_0=\argmax_{i}p_{i1}$ -- see \cite{PMMinf} for details. Recall that $f_j$ denote the emission densities with respect to measure $\mu$. For $j \in \Y$ denote $P_j(A)= \int_{A} f_{j}(x) \, \mu(dx)$. The barrier set $\X^*$ is constructed in such a way that $P_1(\X^*_{(M-r)})>0$ -- see \cite{PMMinf}. For any $A \in \B(\Z)$ and $j \in \Y$ we write $A_j=\{x \:| \: (x,j) \in A\}$. Taking now
\begin{align*}
\nu(A)=\dfrac{P_1(A_{1} \cap \X^*_{(M-r)})}{P_1(\X^*_{(M-r)})},
\end{align*}
we have that \ref{minor} is fulfilled with $\beta=p_{i_0 1}P_1(\X^*_{(M-r)})$. Indeed, for all $x \in \X$
\begin{align*}
P((x,i_0),A)&=\sum_{j \in \Y}p_{i_0j}P_j(A_j)\\
&\geq p_{i_0 1}P_1(A_{1} \cap \X^*_{(M-r)})\\
&=p_{i_0 1}P_1(\X^*_{(M-r)}) \cdot \nu(A).
\end{align*}
Finally note that since $Z$ is $\psi$-irreducible, where
\begin{align*}
\psi(A \times \{i\})=\int_{A}f_i(x) \, \mu(dx), \quad A \in \mathcal{B}(\X), \quad i \in \Y,
\end{align*}
then \ref{R}\ref{Rio} implies $\psi(\Z_0)>0$ (see \cite[Prop. 8.3.7]{MT}). It is not difficult to show that in case of HMM with irreducible $Y$ every set in $\B(\Z)$ is regular and so \ref{recurr} holds.
\end{proof}

\subsection{Discrete $\X$} \label{subsec:disc}
Consider the case where $\X$ is discrete (finite or
countable) and $Z$ is an irreducible and positive recurrent Markov chain with
(discrete) state space ${\cal Z}'\subset {\cal X}\times {\cal
Y}$. Here the state-space refers to the set of possible values of
$Z$. Note that ${\cal Z}'$ can be a proper subset ${\cal X}\times
{\cal Y}$. Recall the definition of $\Y^+(\cdot)$ \eqref{Yplus}. We assume that the transition kernel $q(z|z')$ is defined on ${\cal Z}'$ and so $(i,x_1)\in {\cal Z}'$ for every $i\in \Y^+(x_{1:q})_{(1)}$.
\begin{lemma} \label{discLem} Let $\X$ be discrete and let $Z$ be an
irreducible and positive recurrent Markov chain with the state-space ${\cal
Z}'\subset {\cal X}\times {\cal Y}$. Then the following
conditions ensure that \ref{R}-\ref{recurr} hold.
\begin{enumerate}[label=(\roman*)]
\item \label{discCyc} There exists $n \geq 2$ and $x^*_{1:n} \in \X^n$ such that $(x^*_1,1)\in {\cal Z}'$
and
\begin{enumerate}
\item[1. ] it holds
\begin{align*}
x_1^*=x_n^* \quad \mbox{and} \quad  p_{11}(x^*_{1:n}) > p_{ij}(x^*_{1:n}), \quad \forall i,j \in \mathcal{Y} \setminus \{1\};
\end{align*}
\item[2. ] it holds
\begin{align} \label{discIneq1}
&p_{11}(x^*_{1:n}) > p_{i1}(x^*_{1:n}),\quad \forall i \in \mathcal{Y},\\
&p_{11}(x^*_{1:n}) > p_{1i}(x^*_{1:n}) ,\quad \forall i \in
\mathcal{Y},\label{discIneq2}
\end{align}
where either inequalities \eqref{discIneq1} or inequalities \eqref{discIneq2} could be non-strict.
\end{enumerate}
\item \label{discClust} There exists $q \geq 2$ and a sequence $x_{1:q} \in \X^q$ such that
$\Y^+(x_{1:q})_{(1)}$ is non-empty and $(i,j) \in \Y^+(x_{1:q})$ for
every $i \in \Y^+(x_{1:q})_{(1)}$ and $j \in \Y^+(x_{1:q})_{(2)}$.
\end{enumerate}
\end{lemma}
\begin{proof} It is shown in \cite{PMMinf} that the conditions of the statement imply (PMM1)-(PMM2) with $B_{(n_N-1)}=\{x_{n-1}^*\}$ and $B_{(n_N)}=\{x_n^* \}$. Since $Z$ is positive recurrent on $\Z'$ then with no loss of generality we may assume that all finite sets in $\B(\Z)$ are regular. The statement now follows from Lemma \ref{PMMLem}.
\end{proof}

\subsection{Linear Markov switching model} \label{subsec:linear}
Let $\X= \mathbb{R}^d$ for some $d \geq 1$ and for each state $i \in \Y$ let $\{\xi_k(i)\}_{k \geq 2}$ be an i.i.d. sequence of random variables on $\X$ with $\xi_2(i)$ having density $h_i$ with respect to Lebesgue measure on $\mathbb{R}^d$. We consider the ``linear Markov switching model'', where $X$ is defined recursively by
\begin{align} \label{LMSM}
X_{k}=F(Y_k)X_{k-1}+ \xi_k(Y_k), \quad k \geq 2.
\end{align}
Here $F(i)$ are some $d \times d$ matrices, $Y=\{Y_k\}_{k\geq 1}$ is a Markov chain with transition matrix $(p_{ij})$, $X_1$ is some random variable on $\X$, and random variables $\{\xi_k(i)\}_{k \geq 2, \: i \in \Y}$ are assumed to be independent and independent of $X_1$ and $Y$. Recall that for Markov switching model, the transition density expresses as $q(x,j|x',i)=p_{ij}f_j(x|x')$. For the current model measure $\mu$ is Lebesgue measure on $\mathbb{R}^d$ and $f_j(x|x')=h_j(x-F(j)x')$. When $F(i)$ are zero-matrices, then the linear Markov switching model simply becomes HMM with $h_i$ being the emission densities. When $d=1$, we obtain the ``switching linear autoregression of order 1''. The switching linear autoregressions are popular in econometric modelling, see e.g. \cite{MSM1,MSM2,MSM3,MSM4}.

To show that \ref{R}-\ref{recurr} hold for the linear Markov switching model it suffices to show that the conditions of Lemma \ref{PMMLem} are fulfilled. That $\mu$ is strictly positive is trivially fulfilled in case of Lebesgue measure on $\mathbb{R}^d$. That $(x',x) \mapsto q(x,j|x',i)$ are lower semi-continuous and bounded is fulfilled when $h_i$ are lower semi-continuous and bounded (composition of lower semi-continuous function with continuous function is lower semi-continuous). The conditions for (PMM1) are given by
\begin{lemma} \cite[Lem. 4.2]{PMMinf} \label{LMB3}Let $Z$ be the linear Markov switching model. If the following condition is fulfilled, then $Z$ satisfies (PMM1) with $B_{(n_{N}-1)}$ being an open ball: there exists $x^* \in \X$ such that
\begin{enumerate}[label=(\roman*)] \item \label{LMB3p11}
$p_{11}=\max_{i \in \Y}p_{i1}$;
\item \label{LMB3reach}$(x^*,1)$ is reachable;
\item \label{LMB3cycle}$h_i$ is continuous at $x^*-F(i)x^*$ for all $i \in \Y$, and
\begin{align*}
p_{11}h_1(x^*-F(1)x^*)>p_{ij}h_j(x^*-F(j)x^*), \quad \forall i \in \Y, \quad \forall j \in \Y \setminus \{1\}.
\end{align*}
\end{enumerate}
\end{lemma}
In what follows, for any $x \in \X$ and $r>0$ let $B(x,r)$ denote an open ball in $\X$ with respect to 2-norm with center point $x$ and radius $r>0$. The conditions for (PMM2) are given by

\begin{lemma} \cite[Lem. 4.1]{PMMinf}\label{LMB1B2Lemma} Let $Z$ be the linear Markov switching model. If the following conditions are fulfilled, then $Z$ satisfies (PMM2).
\begin{enumerate}[label=(\roman*)]
\item  \label{LMprimit}There exists set $C\subset \Y$ and $r>0$ such that the following two conditions are satisfied:
\begin{enumerate}
\item[1.] for $x \in B(0,r)$, $h_i(x)>0$ if and only if $i \in C$;
  \item[2.] the sub-stochastic matrix $\mathbb{P}_C= (p_{ij})_{i,j\in C}$ is primitive, i.e. there exists $R \geq 1$ such that matrix $\mathbb{P}_C^R$ has only positive elements.
\end{enumerate}
\item \label{LMreach} Denote $\Y_C=\{i \in \Y \:| \: p_{ij}>0, \: j\in C\}$. There exists $i_E \in \Y_C$ such that $(0,i_E)$ is reachable.
\end{enumerate}
\end{lemma}

Conditions \ref{LMprimit} and \ref{LMreach} are not very restrictive. For example, when all the elements of $\mathbb{P}=(p_{ij})$ are positive, then \ref{LMprimit} is fulfilled if densities $h_i$ are either positive around 0 or zero around 0 and there exists at least one $j \in \Y$ such that $h_j$ is positive around 0. If densities $h_i$ are all positive around 0, then \ref{LMprimit} is fulfilled when $\mathbb{P}$ is primitive with $C=\Y$. If $h_i$ are positive everywhere and $Y$ is irreducible, then all points in $\Z$ are reachable and so \ref{LMreach} trivially holds.

For $x \in \X$ we denote with $\| x \|_1$ the 1-norm of $x$, and for a $d \times d$ matrix $A$ we denote with $\| A \|_1$ the 1-norm of matrix $A$, that is $\|A\|_1$ is the maximum absolute column sum of $A$.  The following lemma addresses Harris recurrence of $Z$ and regularity of sets in $\B(\Z)$. This can be proved analogously to Lemma 4.3 in \cite{PMMinf} by using \cite[Th. 11.3.11]{MT} instead of \cite[Th. 9.1.8]{MT}.

\begin{lemma} \label{LMHarris} Let $Z$ be the linear Markov switching model. If the following conditions are fulfilled, then $Z$ is Harris recurrent and every compact set in $\B(\Z)$ is regular:
\begin{enumerate}[label=(\roman*)] \item $Z$ is $\psi$-irreducible and support of $\psi$ has non-empty interior;
\item $\mathbb{E}[\|\xi_2(i) \|_1]< \infty$ for all $i \in \Y$;
\item $\max_{i \in \Y}\sum_{j \in \Y}p_{ij}\|F(j)\|_1<1$.
\end{enumerate}
\end{lemma}

To recapitulate: if $h_i$ are lower semi-continuous and bounded and the conditions of Lemmas \ref{LMB3}, \ref{LMB1B2Lemma} and \ref{LMHarris} are fulfilled, then the conditions of Lemma \ref{PMMLem} are satisfied and therefore \ref{R}-\ref{recurr} hold. Applying this fact to the case where $h_i$ are Gaussian yields
\begin{cor} \label{GLMcor} Let $Z$ be the linear Markov switching model, with densities $h_i$ being Gaussian with respective mean vectors $\mu_i$ and positive definite covariance matrices $\Sigma_i$. If the following conditions are fulfilled, then \ref{R}-\ref{recurr} hold.
\begin{enumerate}[label=(\roman*)]\item \label{GLMprim} Matrix $\mathbb{P}=(p_{ij})$ is primitive, i.e. there exists $R$ such that $\mathbb{P}^R$ consists of only positive elements.
\item \label{p11Max} It holds $p_{11}=\max_{i \in \Y}p_{i1}$.
\item \label{GLMcycle}Matrix $\mathbb{I}_d - F(1)$, where $\mathbb{I}_d$ denotes the identity matrix of dimension $d$, is non-singular, and for all $i \in \Y$ and $j \in \Y \setminus \{1\}$
\begin{align*}
(\mathbb{I}_d-F(j))(\mathbb{I}_d-F(1))^{-1}\mu_1 \in H_{ij}^{\mathsf{c}},
\end{align*}
where we denote $H_{ij}=\emptyset$, if $p_{ij}=0$ or $\frac{p_{11}\sqrt{|\Sigma_j |}}{p_{ij} \sqrt{|\Sigma_1|}}>1$, and
\begin{align*}
H_{ij} =
\left\{x \in \mathbb{R}^d \: | \: (x- \mu_j)^\top \Sigma_j^{-1}(x- \mu_j) \leq  -2 \ln \left(\dfrac{p_{11}\sqrt{|\Sigma_j |}}{p_{ij} \sqrt{|\Sigma_1|}} \right)  \right\}
\end{align*}
otherwise
\item \label{GLMHarris}It holds $\max_{i \in \Y}\sum_{j \in \Y}p_{ij}\|F(j)\|_1<1$.
\end{enumerate}
\end{cor}
\begin{proof} It is shown in \cite{PMMinf} that under given conditions the assumptions of Lemmas \ref{LMB3}, \ref{LMB1B2Lemma} and \ref{LMHarris} are fulfilled. The statement now follows from Lemma \ref{PMMLem}.
\end{proof}
\section{Asymptotics of Viterbi training: regeneration-based analysis} \label{sec: VT}
Regenerative property of $(Z,V)$ has numerous potential theoretical applications ranging from SLLN's and CLT's to proving the existence of asymptotic risks. In this section we look at one such application, namely regeneration-based analysis of the Viterbi training algorithm.

\paragraph{Convergence of empirical measures}  Firstly, let us state a general SLLN that provides a theoretical
basis for several regeneration-based inferences, including the ones
in the present section. The proof of the theorem is a rather
straightforward application of the regenerative property of $(Z,V)$
and is given in Appendix \ref{PConvProof}.  An analogous result for
HMM was proven in \cite{AV}. Let $V_{1:n}^n$ denote, the Viterbi path of $X_{1:n}$, i.e. $V_{1:n}^n=v(X_{1:n})$. Here the tie-breaking scheme is assumed lexicographic as above. Whenever \ref{R}-\ref{recurr} hold, let $\{S_k\}_{k \geq 1}$ denote the regeneration times of $(Z,V)$, as constructed in the proof of Theorem \ref{reg}.

\begin{theorem}\label{th:fConv} Let \ref{R}-\ref{recurr} hold and let $f \colon (\Z\times \Y) \times (\Z \times \Y) \rightarrow
\mathbb{R}$ be a measurable function satisfying
\begin{equation}\label{eeldus}
\mathbb{E} \left[ \sum_{k=S_{1}+1}^{S_{2}}
 \max_{i,j \in \Y}| f((Z_{k-1},i),(Z_k,j))| \right]<\infty.
 \end{equation}
 Then
\begin{align*}
\dfrac{1}{n-1}\sum_{k=2}^n f((Z_{k-1},V^n_{k-1}),(Z_k,V^n_k)) \xrightarrow[n]{}
 \dfrac{1}{\E[S_2-S_1]}\mathbb{E}\left[\sum_{k=S_{1}+1}^{S_{2}}
f((Z_{k-1},V_{k-1}),(Z_k,V_k)) \right],\quad{\rm
a.s.}
\end{align*}
\end{theorem}

In what follows, we study empirical measures
$$P^n(A) \DEF {1\over n-1}\sum_{k=2}^n  \mathbb{I}_{A }((X_{k-1}, V_{k-1}^n),(X_k, V_k^n)), \quad A \in \B(\Z) \otimes
\B(\Z),$$
where $\mathbb{I}_A$ denotes the indicator function on $A$. Also define measure
\begin{align*}
P^\infty(A)=\dfrac{1}{\E[S_2-S_1]}\mathbb{E}\left[\sum_{k=S_{1}+1}^{S_{2}}
\mathbb{I}_A((X_{k-1},V_{k-1}),(X_k,V_k)) \right], \quad A \in \B(\Z) \otimes \B(\Z).
\end{align*}
From Theorem \ref{th:fConv} it follows that $P^n\Rightarrow P^{\infty}$
a.s., where $\Rightarrow$ stands for the weak convergence of
measures:
\begin{cor}\label{corru} Let \ref{R}-\ref{recurr} hold. If a measurable function $g \colon \Z \times \Z \rightarrow
\mathbb{R}$ satisfies (\ref{eeldus}) with $$f((x_1,y_1, v_1),(x_2,y_2,v_2))=g((x_1,v_1),(x_2,v_2)),$$ then $g$ is $P^\infty$-integrable, and
\begin{align} \label{fConv}
\int g  \, dP^n \xrightarrow[n]{} \int g \, dP^\infty, \quad
\mbox{a.s.}
\end{align}
Moreover, $P^n\Rightarrow P^{\infty}$ a.s.
\end{cor}
\begin{proof} Denote
\begin{align*}
I(g)=\dfrac{1}{\E[S_2-S_1]}\mathbb{E}\left[\sum_{k=S_{1}+1}^{S_{2}}
g((X_{k-1},V_{k-1}),(X_k,V_k)) \right].
\end{align*}
It follows from Theorem \ref{th:fConv} that
 $$\int g \,  dP^n  \xrightarrow[n]{} I(g),\quad a.s.
$$
-- just take $f((x_1,y_1,v_1),(x_2,y_2,v_2))=g((x_1,v_1),(x_2,v_2))$. To prove \eqref{fConv}, we need to show that $I(g)=\int g \, dP^\infty$. This follows by
the standard machinery of measure theory: if $g$ is a simple
function, then the equality $I(g)=\int g \, dP^\infty$ follows by
definition of $P^\infty$. For non-negative measurable $g$ take a
monotone convergent sequence of simple functions $g_m\nearrow g$ and
use Monotone Convergence Theorem to see that the equality $I(g)=\int
g \, dP^\infty$ holds also in the limit. For arbitrary $g$ we have by \eqref{eeldus} that $I(g^+)< \infty$ and $I(g^-)< \infty$, and hence $I(g)=\int g \, dP^\infty$ follows.

The weak convergence now follows from the standard theory of weak convergence of probability measures, see e.g. proof of \cite[Th. 11.4.1]{dudley}.
\end{proof}

Note that when \ref{R}-\ref{recurr} hold, then there exists $\gamma$ such that
\begin{align*}
\dfrac{1}{n-1} \sum_{k=2}^n \mathbb{I}((V_{k-1}^n,V_k^n)\neq (Y_{k-1},Y_k) ) \xrightarrow[]{n} \gamma , \quad \mbox{a.s.,}
\end{align*}
where $\mathbb{I}$ denotes the indicator function. Indeed, this follows directly from Theorem \ref{th:fConv}, by taking
\begin{align*}
f((x_1,y_1,v_1),(x_2,y_2,v_2))= \begin{cases}1, & \mbox{if $(v_1,v_2)\neq(y_1,y_2)$}\\
0, & \mbox{otherwise}
\end{cases}.
\end{align*}
The limit $\gamma$ is simply the asymptotic pairwise misclassification rate of the Viterbi estimation. For any pair of measures $\nu$ and $\nu'$ on common $\sigma$-field, we denote with $\|\nu -\nu'\|$ the total variation distance of measures $\nu$ and $\nu'$:
\begin{align*}
\|\nu -\nu'\|=\sup_{A}|\nu(A)-\nu'(A)|,
\end{align*}
where supremum is taken over all measurable $A$. Recall (Lemma \ref{lem:HarrisPos}) that when \ref{R}-\ref{recurr} hold, then the transition kernel of $Z$ admits an invariant probability measure $\pi$. Let $P_{\pi}$ denote the probability measure with respect to the stationary initial distribution $\pi$. The following result is useful for our purposes.

\begin{prop} \label{propTV} If \ref{R}-\ref{recurr} hold and $Z$ is aperiodic\footnote{See \cite{MT} for the definition of aperiodicity of general state $\psi$-irreducible Markov chain.}, then $$\|P^\infty - P_{\pi}(Z_{1:2} \in \cdot )\| \leq \gamma.$$
\end{prop}
\begin{proof}Note that by  Theorem \ref{th:fConv}
\begin{align} \label{gamEq}
\gamma=\dfrac{1}{\mathbb{E}[S_2-S_1]}\mathbb{E}\left[\sum_{k=S_1+1}^{S_2} \mathbb{I}((V_{k-1},V_k)\neq  (Y_{k-1},Y_k) ) \right]
\end{align}
Similarly, denoting
\begin{align*}
I_k(A)=\mathbb{I}_{A}((X_{k-1},Y_{k-1}),(X_{k},Y_k)), \quad A \in   \B(\Z) \otimes \B(\Z),
\end{align*}
we have
\begin{align} \label{InConv}
\dfrac{1}{n-1}\sum_{k=2}^n I_k(A) \xrightarrow[n]{} \eta(A), \quad \mbox{a.s.,}
\end{align}
where $\eta$ is a probability measure defined by
\begin{align*}
\eta(A)=\dfrac{1}{\mathbb{E}[S_2-S_1]}\mathbb{E}\left[\sum_{k=S_1+1}^{S_2} I_k(A) \right], \quad A \in \B(\Z) \otimes \B(\Z).
\end{align*}
Now we have for every $A \in \B(\Z) \otimes \B(\Z)$
\begin{align*}
\big|\mathbb{I}_{A}((X_{k-1},V_{k-1}),(X_{k},V_k))-I_k(A)\big| \leq \mathbb{I}((V_{k-1},V_k)\neq  (Y_{k-1},Y_k) ), \quad k \geq 1,
\end{align*}
and so by \eqref{gamEq}
\begin{align*}
P^\infty(A) -\eta(A)|&=\dfrac{1}{\mathbb{E}[S_2-S_1]} | \mathbb{E}\left[ \sum_{k=S_1+1}^{S_2} [\mathbb{I}_{A}((X_{k-1},V_{k-1}),(X_{k},V_k))-I_k(A)]\right] | \\
&\leq \gamma.
\end{align*}
By Lemma \ref{lem:HarrisPos} \ref{R}-\ref{recurr} imply that $Z$ is Harris recurrent and positive. To prove the statement, it suffices to show that $\eta=P_{\pi}(Z_{1:2}\in \cdot)$. From the ergodic theorem for aperiodic and positive Harris Markov chains \cite[Th. 13.3.3.]{MT} it follows that for all $A , B \in \B(\Z)$ it holds $\E [I_k(A \times B)] \xrightarrow[k]{} P_{\pi}(Z_{1:2}\in A \times B)$, which implies
\begin{align*}
\dfrac{1}{n-1}\sum_{k=2}^n \E [I_k(A \times B)] \xrightarrow[n]{} P_{\pi}(Z_{1:2}\in A \times B).
\end{align*}
Thus applying Dominated Convergence Theorem to \eqref{InConv} yields $\eta(A \times B)=P_{\pi}(Z_{1:2}\in A \times B)$.
\end{proof}

\paragraph{Viterbi training} The empirical
measures $P^n$ are central in many Viterbi-inferences based
applications. In this section we look at one of such applications,
namely Viterbi training.

For any application of PMM one needs to estimate the parameters of the model. The standard go-to method of parameter estimation in case of Markov switching models is the \textit{EM-algorithm}, also known as \textit{Baum-Welch} or \textit{forward-backward} algorithm \cite{HMMbook,EM1,EM4,MSM1}. In practice the EM-algorithm is often replaced with the \textit{Viterbi training} algorithm (VT) \cite{VT1,VT2,VT3} (also known as \textit{segmental K-means algorithm} \cite{Kmeans1,Kmeans2} and \textit{classification EM} in the mixture case \cite{classEM1,classEM2}) due to its ease of implementation and significantly smaller computational cost. However, compared to the EM-algorithm, VT can sometimes lead to sub-optimal results. It could be argued that the closer the Viterbi path is to the real path, the better VT performs. Thus VT is expected to work well when the dimension of observation space $\X$ is relatively high compared to the number of hidden states, since the Viterbi path is in that case -- in some sense -- probably very similar to the real hidden path. By the same logic, replacing a HMM with a more general PMM (for instance the linear Markov switching model) can improve the accuracy of the Viterbi path estimation and therefore also that of VT. We will see that asymptotic analysis of VT is made possible by the convergence of empirical
measures $P^n$ and, therefore, by regenerativity of $(Z,V)$.

We start by formulation of VT for PMM's. We assume that transition densities $q(z|z')$ are known up to a parametrization $q(z|z')=q(z|z';\theta)$, where $\theta$ is the unknown ``true'' parameter vector belonging to parameter space $\Theta \subset \mathbb{R}^d$, $d \geq 1$. For any parameter vector $\theta' \in \Theta$ denote with $v(x_{1:n};\theta')$ the Viterbi path of $x_{1:n}$ obtained by using the parameters $\theta'$. Hence $v(x_{1:n})=v(x_{1:n}; \theta)$. Viterbi training is a method of estimating the unknown parameters using the following algorithm.
\bigskip

{\noindent \bf Viterbi training algorithm:}
\begin{enumerate}
\item Choose $^0\theta \in \Theta$, the initial values for the parameter vector, and set $l=0$.
\item Using current parameters ${}^l\theta$, obtain the corresponding Viterbi path ${}^lv={}^lv_{1:n}=v(x_{1:n};{}^l\theta)$.
\item Define empirical measure
\begin{align*}
P^n(A; x_{1:n}, {}^l\theta)=\dfrac{1}{n-1}\sum_{k=2}^n \mathbb{I}_{A }((x_{k-1},{}^lv_{k-1}),(x_k,{}^lv_k)), \quad
 A \in \B(\Z) \otimes \B(\Z)
\end{align*}
and update the parameters ${}^l\theta$ by taking
\begin{align*}
{}^{l+1}\theta&=\argmax_{\theta' \in \Theta} \int \ln q(z|z'; \theta') \, P^n(d(z',z); x_{1:n}, {}^l\theta).
\end{align*}
\item Put $l=l+1$ and repeat from step 2.
\end{enumerate}

As a special case, we can consider the Markov switching model defined by transition kernel density $p_{ij}f_{j}(x|x';\alpha_{j})$, where $\theta=(p_{ij},\alpha_{j})_{i, j \in \Y}$ is the unknown parameter vector. Then (as one can see by applying Lagrange multipliers method) the estimates of ${}^{l+1}p_{ij}$ of $p_{ij}$ at the 3rd step of Viterbi training algorithm can simply be calculated by
\begin{align*}
{}^{l+1}p_{ij} = \begin{cases} \dfrac{\sum_{k=2}^n \mathbb{I}_{\{(i,j)\}}({}^lv_{k-1},{}^lv_k)}{\sum_{k=2}^n \mathbb{I}_{\{i\}}({}^lv_{k-1})}, & \mbox{if $\sum_{k=2}^n \mathbb{I}_{\{i\}}({}^lv_{k-1})>0$},\\
{}^{0}p_{ij}, & \mbox{otherwise}
\end{cases}.
\end{align*}
Then also ${}^{l+1}\alpha_{j}$, the estimates of $\alpha_{j}$ at the 3rd step of Viterbi training algorithm, can be calculated separately by
\begin{align*}
{}^{l+1}\alpha_{j}&=\argmax_{\alpha \in \mathcal{A}_{j}} \sum_{k=2}^n \mathbb{I}_{\{j\}}(^l v_k) \ln f_{j}(x_k|x_{k-1}; \alpha).
\end{align*}
where $\mathcal{A}_{j}$ is the parameter space corresponding to $\alpha_{j}$. Recall that Markov switching model also includes HMM in which case $f_j(x|x'; \alpha_j)$ is independent of $x'$.


Assume for now that \ref{R}-\ref{recurr} hold and $Z$ is aperiodic. In the light of Theorem \ref{th:fConv} we may want to seek the answer to the following question: what is the output of the 3rd step of the Viterbi training algorithm if the input is the correct parameter $\theta$? Ideally, we would like to see that that the output is in that case close to the input: $\theta={}^l\theta \approx {}^{l+1}\theta $. This is called the \textit{fixed point property}. Write $q(\theta')$ for $(z',z) \mapsto q(z|z';\theta')$. Assuming that $q(\theta')$ is positive everywhere on $\Z^2$ and $\ln q(\theta')$ satisfies \eqref{eeldus} we have by Corollary \ref{corru} that
\begin{align*}
\int \ln q(\theta') \,  dP^n(\cdot \, ; X_{1:n}, \theta) \xrightarrow[n]{} \int \ln  q(\theta') \, dP^\infty , \quad \mbox{a.s.}
\end{align*}
This gives hope that at least under some conditions
\begin{align} \label{argMConv}
\argmax_{\theta' \in \Theta}\int \ln  q(\theta') \,  dP^n(\cdot \, ; X_{1:n}, \theta) \xrightarrow[n]{} \argmax_{\theta' \in \Theta}\int \ln q(\theta') \,    dP^\infty, \quad \mbox{a.s.}
\end{align}
In \cite{AVacta} an analogous convergence is proved under certain conditions. Although \cite{AVacta} deals with HMM's only, we believe that the same idea can be applied to PMM as well to obtain the convergence \eqref{argMConv}. Denote
\begin{align*}
\theta^*(\theta)=\argmax_{\theta' \in \Theta}\int \ln q(\theta') \,    dP^\infty.
\end{align*}
Assuming that \eqref{argMConv} holds and that the sample size $n$ is sufficiently large, the fixed point property holds when $\theta^* \approx \theta$. It is not difficult to confirm that $P^\infty$ has a density with respect to measure $(\mu \times c)\times (\mu \times c)$, where $c$ denotes counting measure on $2^\Y$. Let $p'(z',z)$ denote this density.  Writing $d(z',z)$ in place of $(\mu \times c)\times (\mu \times c)(d(z',z))$ we have
\begin{align*}
\theta^*= \argmax_{\theta' \in \Theta}\int p'(z',z) \ln q(z|z'; \theta') \,    d(z',z).
\end{align*}
It is not difficult to confirm that the invariant measure $\pi$ has a density with respect to measure $\mu \times c$ -- denote this density with $p_{\pi}$. Assuming that $\ln p'(z',z)$ and $(z',z) \mapsto \ln p_{\pi}(z')$ are both $P^\infty$-integrable, we have
\begin{align}
\theta^* &= \argmax_{\theta' \in \Theta}\int p'(z',z) [\ln q(z|z'; \theta') + \ln p_{\pi}(z')- \ln p'(z',z)] \,    d(z',z) \notag\\
&= \argmin_{\theta' \in \Theta} \int p'(z',z)  \ln \left(\dfrac{p'(z', z)}{p_{\pi}(z')q(z|z'; \theta')}\right)\,    d(z',z). \label{argminKL}
\end{align}
Defining measure $Q_{\theta'}$ by
\begin{align*}
Q_{\theta'}(A)=\int_A p_{\pi}(z')q(z|z'; \theta') \, d(z',z), \quad A \in \B(\Z) \otimes \B(\Z),
\end{align*}
we can see that \eqref{argminKL} minimizes the Kullback-Leibler divergence of $P^\infty$ from $Q_{\theta'}$. If the Viterbi path is perfect estimate of the hidden path, i.e. $V \equiv Y$, then by Proposition \ref{propTV} $P^\infty =P_\pi(Z_{1:2} \in \cdot)=Q_{\theta}$. The Kullback-Leibler divergence of $Q_{\theta}$ from $Q_{\theta'}$ is minimal when $\theta'=\theta$, implying that the fixed point property $\theta^* \approx \theta$ indeed holds when the Viterbi path is perfect, at least under appropriate identifiability conditions. In practice however the Viterbi path is almost never perfect and so whether $\theta^* \approx \theta$ holds or not depends on how closely $P^\infty$ resembles $P_\pi(Z_{1:2} \in \cdot)$. If the pairwise misclassification rate $\gamma$ is small, then by Proposition \ref{propTV} $\|P^\infty - P_{\pi}(Z_{1:2}\in \cdot)\| \leq \gamma$, and so the fixed point property is in that case expected to hold.

If $P^\infty$ is not close to $P_\pi(Z_{1:2} \in \cdot)$, we can correct the estimation error at the 3rd step of the VT by taking the estimate ${}^{l+1}\theta$ to be
\begin{align*}
{}^{l+1}\theta &=\theta- \theta^*(\theta)+\argmax_{\theta' \in \Theta} \int \ln q(z|z'; \theta') \, P^n(d(z',z); X_{1:n}, {}^l\theta).
\end{align*}
In case of HMM this modified version of the VT has been called the \textit{adjusted Viterbi training} (AV) \cite{AV,Eng,AVacta}. AV ensures that the fixed point property holds simply by virtue of convergence \eqref{argMConv}. The simulations in \cite{Eng,AVacta} show how the fixed point property significantly
improves the accuracy of the estimator so that AV outperforms VT and
is comparable to EM-algorithm. The key component of AV is the map $\theta \mapsto \theta^*(\theta)$. This map is generally not known analytically and should therefore be estimated by simulations. The estimation takes some time and effort but could be worthwhile if a single model is to be applied repeatedly.

\paragraph{Simulations} Based on the arguments above it should be believable that the pairwise misclassification rate $\gamma$ provides a fairly adequate measure on how well VT performs in practice (smaller values of $\gamma$ indicating a better performance). This rate also has the advantage of being easy to calculate in simulations. Recall now the linear Markov switching defined by \eqref{LMSM}. We consider the case where $|\Y|=2$, $\X = \mathbb{R}$, $F(1)=0$, $p_{11}=p_{22}=0.9$, $h_1$ is standard normal and $h_2$ is normal with mean $\mu$ and standard deviation 1. It follows from Corollary \ref{GLMcor} that \ref{R}-\ref{recurr} hold whenever $\mu \neq 0$ and $|F(2)|<1 \frac{1}{9}$.

Figure \ref{cont} depicts the estimated value of $\gamma$ depending on $\mu$ and $F(2)$. In case of $\mu=F(2)=0$ Viterbi path is constantly 1 (due to lexicographic ordering) and all non-$(1,1)$ pairs, which constitute slightly more than half of all the pairs, are misclassified. As $\mu$ and $F(2)$ increase, the misclassification rate $\gamma$ decreases. The case $F(2)=0$ corresponds to HMM. From Figure
\ref{cont} we see that as $F(2)$ increases, i.e. the
model ``moves away from HMM'', $\gamma$ decreases.

\begin{figure}
    \centering
        \includegraphics[width=0.8\textwidth, trim={0 0.5cm 0 0.5cm},clip]{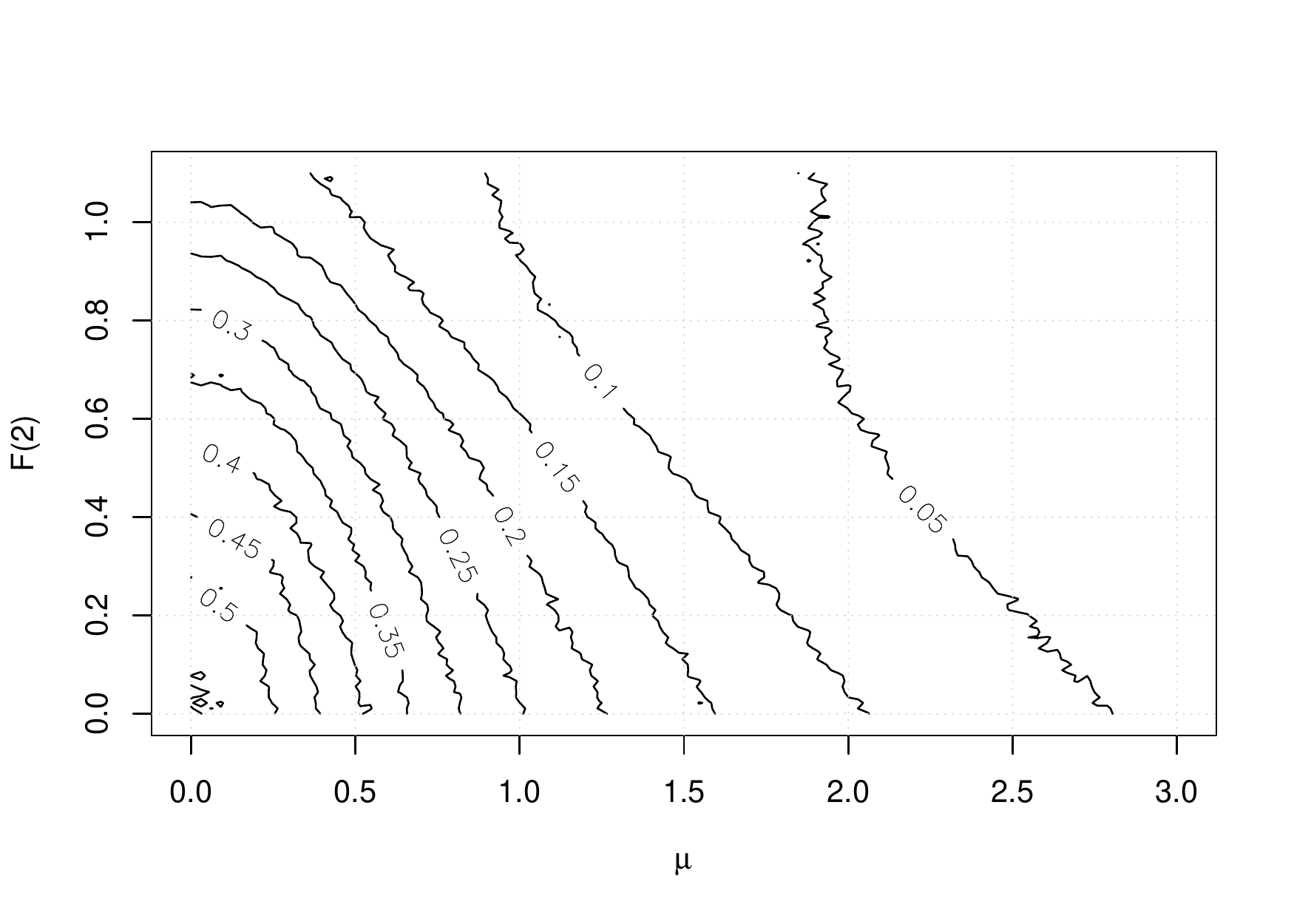}
    \caption{Estimated value of $\gamma$ depending on model parameters $\mu$ and $F(2)$.} \label{cont}
\end{figure}
Let $(X_1^\infty, V_1^\infty, X_2^\infty,V_2^\infty)$ be a random vector having distribution $P^\infty$. Figure \ref{2dPlot} depicts the densities of $P(X_{1:2}^\infty \in \cdot | V^\infty_{1:2}=(i,j))$, drawn with solid line, and $P_\pi(X_{1:2} \in \cdot| Y_{1:2}=(i,j))$, drawn with dotted line. The upper left, upper right, bottom left and bottom right plot represent pairs $(1,1), (1,2), (2,1), (2,2)$, respectively. In Figure \ref{2d1} $F(2)=0$ (i.e. the model is HMM) and $\mu=1$. In Figure \ref{2d2} $F(2)=\frac{1}{2}$ and $\mu=2$. In the HMM-case the density of $P_\pi(X_{1:2} \in \cdot| Y_{1:2}=(i,j))$ is $h_i(x_1)h_j(x_2)$ and can thus be directly calculated; all other densities are estimated by simulation. In the HMM-case (Figure \ref{2d1}) the difference of densities $P(X_{1:2}^\infty \in \cdot | V^\infty_{1:2}=(i,j))$ and $P_\pi(X_{1:2} \in \cdot| Y_{1:2}=(i,j))$ seem to be insignificant for $(i,j)=(1,1),(2,2)$ and greater for $(i,j)=(1,2),(2,1)$. In the PMM-case (Figure \ref{2d2}) the difference between two densities is still very small for $(i,j)=(1,1),(2,2)$ while for pairs $(i,j)=(1,2),(2,1)$ the densities appear to be closer than in the HMM-case. This accords with the fact that in the HMM-case -- as can be seen from Figure \ref{cont} -- the value of $\gamma$ is about 0.25 while in the PMM-case $\gamma$ is much smaller -- roughly about 0.05.
\begin{figure}
    \centering
    \begin{subfigure}[b]{0.47\textwidth}
        \includegraphics[width=\textwidth]{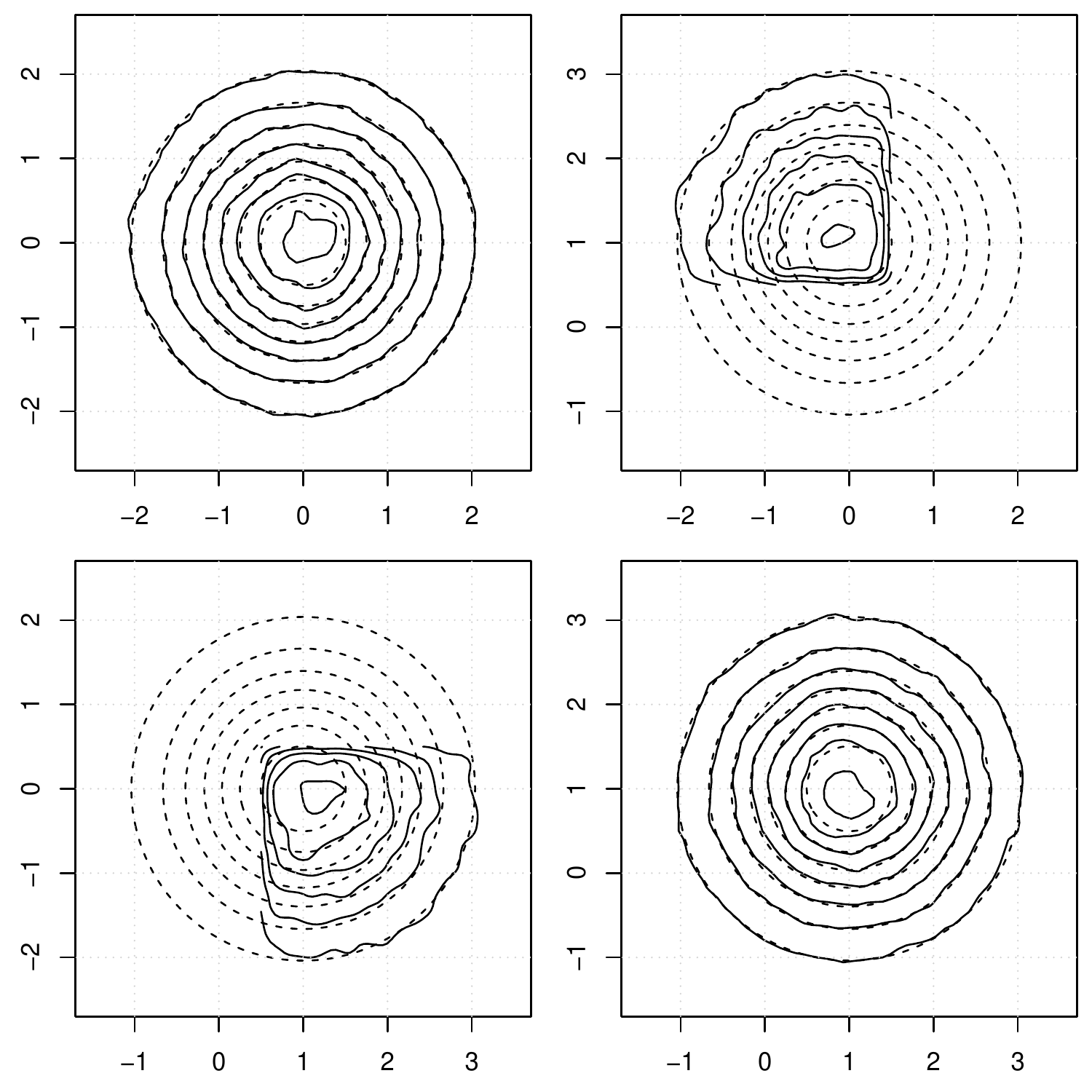}
        \caption{} \label{2d1}
    \end{subfigure}
    ~~
    \begin{subfigure}[b]{0.47\textwidth}
        \includegraphics[width=\textwidth]{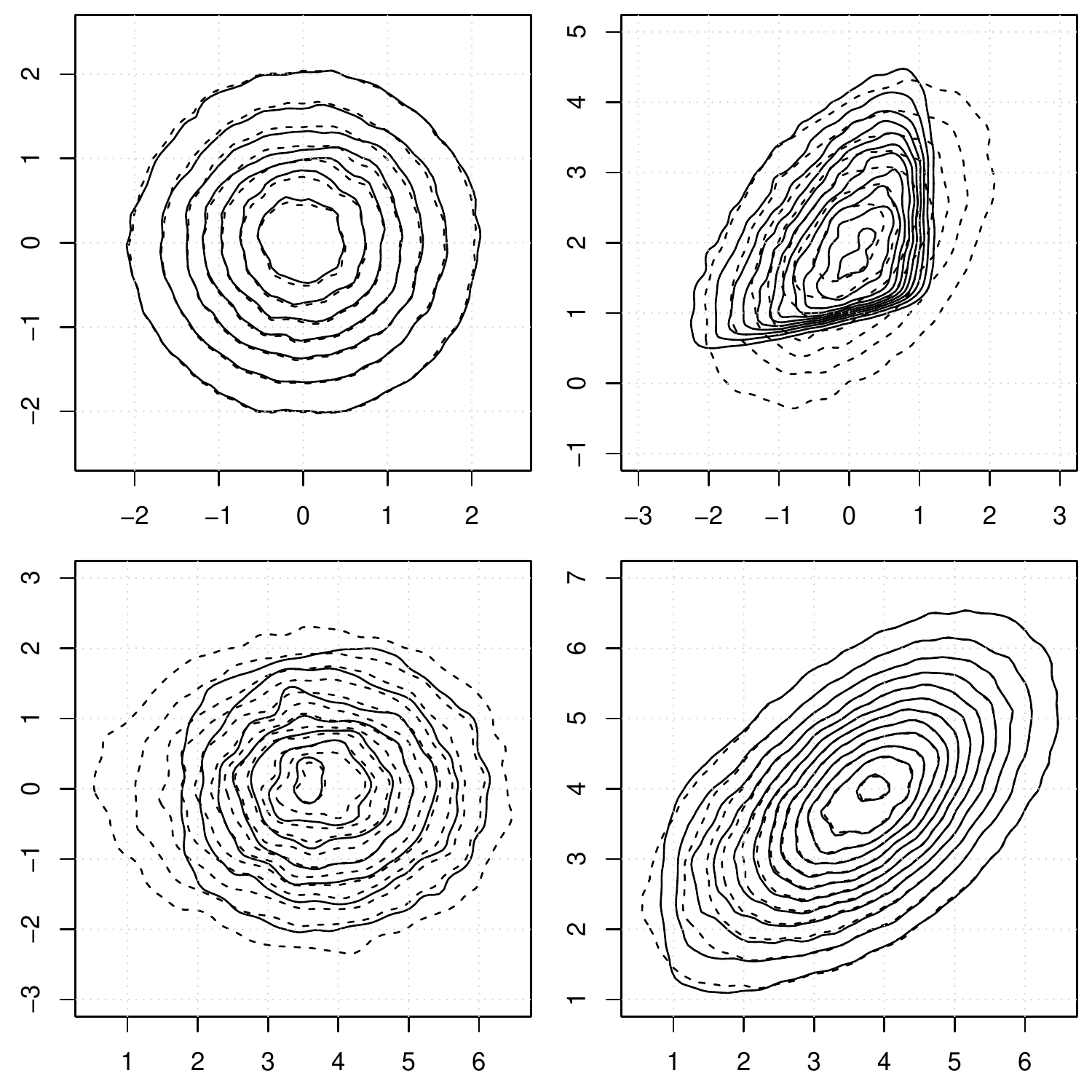}
        \caption{} \label{2d2}
    \end{subfigure}
    \caption{Densities of $P(X_{1:2}^\infty \in \cdot | V^\infty_{1:2}=(i,j))$, drawn with solid line, and $P_\pi(X_{1:2} \in \cdot| Y_{1:2}=(i,j))$, drawn with dotted line. The upper left, upper right, bottom left and bottom right plot represent pairs $(1,1), (1,2), (2,1), (2,2)$, respectively. In (a) $F(2)=0$ (i.e. the model is HMM) and $\mu=1$. In (b) $F(2)=\frac{1}{2}$ and $\mu=2$.} \label{2dPlot}
\end{figure}
\begin{appendices}
\section{Proof of Theorem \ref{th:fConv}} \label{PConvProof} Let $f \colon (\Z\times \Y) \times (\Z \times \Y) \rightarrow \mathbb{R}$ be a measurable function satisfying \eqref{eeldus}. Denote
\begin{align*}
f^n(k)=f((Z_{k-1},V^n_{k-1}), (Z_k ,V^n_k)), \quad k \geq 2, \quad n \geq k,
\end{align*}
and
\begin{align*}
f(k)=f((Z_{k-1},V_{k-1}), (Z_k ,V_k)), \quad k \geq 2.
\end{align*}
We need to show that
\begin{align*}
\dfrac{1}{n-1}\sum_{k=2}^n f^n(k) \xrightarrow[n]{} \dfrac{1}{\mathbb{E}[S_2-S_1]} \E \left[\sum_{k=S_1+1}^{S_2}f(k) \right], \quad \mbox{a.s.}
\end{align*}
Denote $K(n)=\max\{k \: | \: S_k \leq n\}$. The regeneration times $S_k$ and Viterbi process $V$ are constructed in such a way that $V^n_{1:S_{K(n)-1}} = V_{1:S_{K(n)-1}}$  for each $n \geq S_2$. Thus for $n \geq S_3$ we have
\begin{align} \label{3terms}
\dfrac{1}{n-1}\sum_{k=2}^n f^n(k)=\dfrac{1}{n-1} \sum_{k=2}^{S_1} f(k) + \dfrac{1}{n-1} \sum_{k=S_1+1}^{S_{K(n)-1}} f(k) +\dfrac{1}{n-1} \sum_{k=S_{K(n)-1}+1}^{n} f^n(k).
\end{align}
We will take a closer look at the three terms on the right side of \eqref{3terms}. Since $S_1< \infty$ a.s., then the first term converges to zero (a.s.) as $n \rightarrow \infty$.

Next we show that the third term converges to zero as well. Denote
\begin{align*}
g(l)=\sum_{k=S_{l}+1}^{S_{l+1}}\max_{i,j \in \Y}|f((Z_{k-1},i), (Z_k,j))|, \quad l \geq 1.
\end{align*}
The regeneration times $S_k$ split $(Z,V)$ into i.i.d. cycles, so $\{g(2l-1)\}_{l \geq 1}$ and $\{g(2l)\}_{l \geq 1}$ are both i.i.d. ($\{g(l)\}_{l \geq 1}$ is not generally i.i.d. because $g(l)$ and $g(l+1)$ are not independent). By \eqref{eeldus} $\mathbb{E}[g(1)]< \infty$, so by SLLN we have
\begin{align*}
\dfrac{1}{l}\sum_{k=1}^l g(k)&= \dfrac{\ceil{\frac{l}{2}}}{l} \dfrac{1}{\ceil{\frac{l}{2}}} \sum\limits_{\substack{k=1, \ldots,l \\ \mbox{\scriptsize{$k$ is odd}}}}g(k) +\dfrac{\floor{\frac{l}{2}}}{l} \dfrac{1}{\floor{\frac{l}{2}}}\sum\limits_{\substack{k=1, \ldots,l \\ \mbox{\scriptsize{$k$ is even}}}}g(k)  \xrightarrow[l]{\mbox{\scriptsize{a.s.}}} \dfrac{1}{2}\E[g(1)] +\dfrac{1}{2}\E[g(1)]\\
&= \E[g(1)].
\end{align*}
This implies that $\frac{1}{l}g(l)$ converges to zero as $l$ increases. Since $K(n)\nearrow \infty$, a.s.,  we have
\begin{align*}
\dfrac{1}{n-1} \sum_{k=S_{K(n)-1}+1}^{n} f^n(k) \leq \dfrac{1}{K(n)-1}(g(K(n)-1) +g(K(n)) \xrightarrow[n]{\mbox{\scriptsize{a.s.}}} 0,
\end{align*}
as claimed.

Denote
\begin{align*}
h(l)=\sum_{k=S_{l}+1}^{S_{l+1}}f((Z_{k-1},V_{k-1}), (Z_k,V_k))|, \quad l \geq 1.
\end{align*}
It remains to show that the second term on the right side of \eqref{3terms} converges to $\frac{1}{\E[S_2-S_1]}\mathbb{E}[h(1)]$. We have by SLLN almost surely
\begin{align*}
\lim_n \dfrac{S_{K(n)}}{K(n)} = \lim_n \dfrac{S_{K(n)+1}}{K(n)}= \E[S_2-S_1],
\end{align*}
and so inequalities $S_{K(n)} \leq n \leq S_{K(n)+1}$ imply convergence
\begin{align*}
\dfrac{n}{K(n)} \xrightarrow[n]{} \E[S_2-S_1], \quad \mbox{a.s.}
\end{align*}
Note that $\E[h(1)] \leq \E[g(1)]< \infty$. By regenerativity $\{h(2l-1)\}_{l \geq 1}$ and $\{h(2l)\}_{n \geq 1}$ are both i.i.d., so by SLLN
\begin{align*}
 \dfrac{1}{n-1} \sum_{k=S_1+1}^{S_{K(n)-1}} f(k) &=
 \dfrac{\ceil{\frac{K(n)-2}{2}}}{n-1} \dfrac{1}{\ceil{\frac{K(n)-2}{2}}}\sum\limits_{\substack{k=1, \ldots,K(n)-2 \\ \mbox{\scriptsize{$k$ is odd}}}} h(k)
  \quad +
 \dfrac{\floor{\frac{K(n)-2}{2}}}{n-1} \dfrac{1}{\floor{\frac{K(n)-2}{2}}}\sum\limits_{\substack{k=1, \ldots,K(n)-2 \\ \mbox{\scriptsize{$k$ is even}}}} h(k) \\
& \xrightarrow[n]{\mbox{\scriptsize{a.s.}}} \dfrac{1}{2\E[S_2-S_1]}\E[h(1)] +\dfrac{1}{2\E[S_2-S_1]} \E[h(1)]\\
&=\dfrac{1}{\E[S_2-S_1]}\E[h(1)].
\end{align*}
\end{appendices}

\bibliographystyle{abbrvnat}
\bibliography{Reg}

\begin{thebibliography}{38}
\providecommand{\natexlab}[1]{#1}
\providecommand{\url}[1]{\texttt{#1}}
\expandafter\ifx\csname urlstyle\endcsname\relax
  \providecommand{\doi}[1]{doi: #1}\else
  \providecommand{\doi}{doi: \begingroup \urlstyle{rm}\Url}\fi

\bibitem[Baum and Petrie(1966)]{EM1}
L.~E. Baum and T.~Petrie.
\newblock Statistical inference for probabilistic functions of finite state
  markov chains.
\newblock \emph{The annals of mathematical statistics}, 37\penalty0
  (6):\penalty0 1554--1563, 1966.

\bibitem[Caliebe(2006)]{caliebe2}
A.~Caliebe.
\newblock Properties of the maximum a posteriori path estimator in hidden
  {M}arkov models.
\newblock \emph{IEEE Transactions on Information Theory}, 52\penalty0
  (1):\penalty0 41--51, 2006.

\bibitem[Caliebe and R{\"o}sler(2002)]{caliebe1}
A.~Caliebe and U.~R{\"o}sler.
\newblock Convergence of the maximum a posteriori path estimator in hidden
  {M}arkov models.
\newblock \emph{IEEE Transactions on Information Theory}, 48\penalty0
  (7):\penalty0 1750--1758, 2002.

\bibitem[Capp{\'e} et~al.(2005)Capp{\'e}, Moulines, and Ryd{\'e}n]{HMMbook}
O.~Capp{\'e}, E.~Moulines, and T.~Ryd{\'e}n.
\newblock \emph{Inference in hidden {M}arkov models}.
\newblock Springer, 2005.

\bibitem[Celeux and Govaert(1992)]{classEM1}
G.~Celeux and G.~Govaert.
\newblock A classification em algorithm for clustering and two stochastic
  versions.
\newblock \emph{Computational statistics \& Data analysis}, 14\penalty0
  (3):\penalty0 315--332, 1992.

\bibitem[Chigansky and Ritov(2011)]{Ritov}
P.~Chigansky and Y.~Ritov.
\newblock On the {V}iterbi process with continuous state space.
\newblock \emph{Bernoulli}, 17\penalty0 (2):\penalty0 609--627, 2011.

\bibitem[Derrode and Piecynski(2004)]{pairwise2}
S.~Derrode and W.~Piecynski.
\newblock Signal and image segmentation using pairwise {M}arkov chains.
\newblock \emph{IEEE Transactions on Signal Processing}, 52\penalty0
  (9):\penalty0 2477--2489, 2004.

\bibitem[Derrode and Piecynski(2013)]{pairwise3}
S.~Derrode and W.~Piecynski.
\newblock Unsupervised data classification using pairwise {M}arkov chains with
  automatic copula selection.
\newblock \emph{Computational Statistics and Data Analysis}, 63:\penalty0
  81--98, 2013.

\bibitem[Dudley(2002)]{dudley}
R.~M. Dudley.
\newblock \emph{Real analysis and probability}, volume~74.
\newblock Cambridge University Press, 2002.

\bibitem[Ephraim and Merhav(2002)]{Kmeans1}
Y.~Ephraim and N.~Merhav.
\newblock Hidden markov processes.
\newblock \emph{IEEE Transactions on information theory}, 48\penalty0
  (6):\penalty0 1518--1569, 2002.

\bibitem[Fraley and Raftery(2002)]{classEM2}
C.~Fraley and A.~E. Raftery.
\newblock Model-based clustering, discriminant analysis, and density
  estimation.
\newblock \emph{Journal of the American statistical Association}, 97\penalty0
  (458):\penalty0 611--631, 2002.

\bibitem[Ghosh et~al.(2011)Ghosh, Kleiman, and Roitershtein]{iowa}
A.~Ghosh, E.~Kleiman, and A.~Roitershtein.
\newblock Large deviation bounds for functionals of {V}iterbi paths.
\newblock \emph{IEEE Transactions on Information Theory}, 57\penalty0
  (6):\penalty0 3932--3937, 2011.

\bibitem[Goodwin(1993)]{MSM4}
T.~H. Goodwin.
\newblock Business-cycle analysis with a markov-switching model.
\newblock \emph{Journal of Business \& Economic Statistics}, 11\penalty0
  (3):\penalty0 331--339, 1993.

\bibitem[Hamilton(1989)]{MSM1}
J.~D. Hamilton.
\newblock A new approach to the economic analysis of nonstationary time series
  and the business cycle.
\newblock \emph{Econometrica: Journal of the Econometric Society}, pages
  357--384, 1989.

\bibitem[Hamilton(1990)]{MSM2}
J.~D. Hamilton.
\newblock Analysis of time series subject to changes in regime.
\newblock \emph{Journal of econometrics}, 45\penalty0 (1-2):\penalty0 39--70,
  1990.

\bibitem[Hamilton(2010)]{MSM3}
J.~D. Hamilton.
\newblock Regime switching models.
\newblock In \emph{Macroeconometrics and time series analysis}, pages 202--209.
  Springer, 2010.

\bibitem[Huang et~al.(1990)Huang, Ariki, and Jack]{VT2}
X.~D. Huang, Y.~Ariki, and M.~A. Jack.
\newblock \emph{Hidden Markov models for speech recognition}, volume 2004.
\newblock Edinburgh university press Edinburgh, 1990.

\bibitem[Juang and Rabiner(1990)]{Kmeans2}
B.-H. Juang and L.~R. Rabiner.
\newblock The segmental k-means algorithm for estimating parameters of hidden
  markov models.
\newblock \emph{IEEE Transactions on Acoustics, Speech, and Signal Processing},
  38\penalty0 (9):\penalty0 1639--1641, 1990.

\bibitem[Kalashnikov(1994)]{kalashnikov}
V.~V. Kalashnikov.
\newblock \emph{Topics on regenerative processes}.
\newblock CRC Press, 1994.

\bibitem[Koloydenko and Lember(2008)]{K2}
A.~Koloydenko and J.~Lember.
\newblock Infinite {V}iterbi alignments in the two state hidden {M}arkov
  models.
\newblock \emph{Acta et Commentationes Universitatis Tartuensis de
  Mathematica}, 12:\penalty0 109--124, 2008.

\bibitem[Koloydenko and Lember(2014)]{seg}
A.~Koloydenko and J.~Lember.
\newblock Bridging {V}iterbi and posterior decoding: A generalized risk
  approach to hidden path inference based on hidden {M}arkov models.
\newblock \emph{Journal of Machine Learning Research}, 15:\penalty0 1--58,
  2014.

\bibitem[Koloydenko et~al.(2007)Koloydenko, K{\"a}{\"a}rik, and Lember]{AVacta}
A.~Koloydenko, M.~K{\"a}{\"a}rik, and J.~Lember.
\newblock On adjusted {V}iterbi training.
\newblock \emph{Acta Applicandae Mathematicae}, 96\penalty0 (1):\penalty0
  309--326, 2007.

\bibitem[Kuljus and Lember(2012)]{Vrisk}
K.~Kuljus and J.~Lember.
\newblock Asymptotic risks of {V}iterbi segmentation.
\newblock \emph{Stochastic Processes and their Applications}, 122\penalty0
  (9):\penalty0 3312--3341, 2012.

\bibitem[Kuljus and Lember(2016)]{peep}
K.~Kuljus and J.~Lember.
\newblock On the accuracy of the {MAP} inference in {HMM}s.
\newblock \emph{Methodology and Computing in Applied Probability}, 18\penalty0
  (3):\penalty0 597--627, 2016.

\bibitem[Lember(2011)]{Vsmoothing}
J.~Lember.
\newblock On approximation of smoothing probabilities for hidden {M}arkov
  models.
\newblock \emph{Statistics \& probability letters}, 81\penalty0 (2):\penalty0
  310--316, 2011.

\bibitem[Lember and Koloydenko(2007)]{Eng}
J.~Lember and A.~Koloydenko.
\newblock Adjusted {V}iterbi training. {A} proof of concept.
\newblock \emph{Probability in the Engineering and Informational Sciences},
  21\penalty0 (3):\penalty0 451–475, 2007.

\bibitem[Lember and Koloydenko(2010)]{AVT5}
J.~Lember and A.~Koloydenko.
\newblock A constructive proof of the existence of {V}iterbi processes.
\newblock \emph{IEEE Transactions on Information Theory}, 56\penalty0
  (4):\penalty0 2017--2033, 2010.

\bibitem[Lember and Sova(2017)]{PMMinf}
J.~Lember and J.~Sova.
\newblock Existence of infinite viterbi path for pairwise markov models.
\newblock \emph{arXiv preprint arXiv:1708.03799}, 2017.

\bibitem[Lember et~al.(2008)Lember, Koloydenko, et~al.]{AV}
J.~Lember, A.~Koloydenko, et~al.
\newblock The adjusted {V}iterbi training for hidden {M}arkov models.
\newblock \emph{Bernoulli}, 14\penalty0 (1):\penalty0 180--206, 2008.

\bibitem[Lember et~al.(2011)Lember, Kuljus, and Koloydenko]{intech}
J.~Lember, K.~Kuljus, and A.~Koloydenko.
\newblock {Theory of segmentation}.
\newblock In P.~Dymarsky, editor, \emph{Hidden Markov Models, Theory and
  Applications}, pages 51--84. InTech, 2011.

\bibitem[McDermott and Hazen(2004)]{VT3}
E.~McDermott and T.~J. Hazen.
\newblock Minimum classification error training of landmark models for
  real-time continuous speech recognition.
\newblock In \emph{Acoustics, Speech, and Signal Processing, 2004.
  Proceedings.(ICASSP'04). IEEE International Conference on}, volume~1, pages
  I--937. IEEE, 2004.

\bibitem[Meyn and Tweedie(2009)]{MT}
S.~P. Meyn and R.~Tweedie.
\newblock \emph{Markov Chains and Stochastic Stability}.
\newblock Cambridge University Press, 2009.

\bibitem[Pieczynski(2003)]{pairwise}
W.~Pieczynski.
\newblock Pairwise {M}arkov chains.
\newblock \emph{IEEE Transactions on Pattern Analysis and Machine
  Intelligence}, 25\penalty0 (5):\penalty0 634--639, 2003.

\bibitem[Rabiner(1989)]{EM4}
L.~R. Rabiner.
\newblock A tutorial on hidden markov models and selected applications in
  speech recognition.
\newblock \emph{Proceedings of the IEEE}, 77\penalty0 (2):\penalty0 257--286,
  1989.

\bibitem[Rodr{\'\i}guez and Torres(2003)]{VT1}
L.~Rodr{\'\i}guez and I.~Torres.
\newblock Comparative study of the baum-welch and viterbi training algorithms
  applied to read and spontaneous speech recognition.
\newblock \emph{Pattern Recognition and Image Analysis}, pages 847--857, 2003.

\bibitem[{\v{S}}r{\'a}mek et~al.(2007){\v{S}}r{\'a}mek, Brejov{\'a}, and
  Vinar]{OnLine}
R.~{\v{S}}r{\'a}mek, B.~Brejov{\'a}, and T.~Vinar.
\newblock On-line viterbi algorithm for analysis of long biological sequences.
\newblock \emph{Algorithms in Bioinformatics}, 4645:\penalty0 240--251, 2007.

\bibitem[Thorisson(2000)]{thorisson}
H.~Thorisson.
\newblock \emph{Coupling, stationarity, and regeneration}, volume 200.
\newblock Springer New York, 2000.

\bibitem[Yau and Holmes(2013)]{chris}
C.~Yau and C.~Holmes.
\newblock A decision-theoretic approach for segmental classification.
\newblock \emph{Annals of Applied Statistics}, 7\penalty0 (3):\penalty0
  1814--1835, 2013.

\end{thebibliography}
\end{document}